\documentclass[twoside,11pt]{article}
\usepackage{jmlr2e}

\usepackage{caption}
\usepackage{color}
\usepackage{hyperref}
\hypersetup{
    colorlinks = true, 
    linkcolor = blue,
    filecolor = magenta, 
    citecolor = green,
    urlcolor = cyan,
    pdfborder = {0 0 0}
}
\usepackage{algorithm}
\usepackage[noend]{algorithmic}
\newtheorem{theorem}{Theorem}

\newtheorem{definition}{Definition}

\usepackage{booktabs}
\usepackage{amsmath} 
\usepackage{enumitem}
\usepackage{bm}

\newcommand{\cgame}{$C^{3}_{\text{ext}}$}

\usepackage{subcaption}

\def\min{\qopname\relax n{min}}
\def\max2{\qopname\relax n{max2}}
\def\max{\qopname\relax n{max}}

\newcommand{\EE}{\mathbb{E}}

\newcommand{\RR}{\mathbb{R}}
\newcommand{\NN}{\mathbb{N}}

\def\G{\mathcal{G}}

\def\N{\mathcal{N}}

\def\S{\mathcal{S}}

\def\X{\mathcal{X}}

\def \cG {\mathcal{G}}

\def\x{\bm{x}} 
 
\def\y{\bm{y}} 
\def\w{\mathbf{w}} 
 
\def\c{\bm{c}} 
\def\e{\mathbf{e}}

\def\t{\bm{t}}

\def\e{\bm{e}} 
 
\def\g{\bm{g}} 
 
\def\s{\bm{s}}

\def\w{\bm{w}} 
\def\z{\bm{z}}

\newenvironment{lp*}{\begin{equation*}  \begin{array}{lll}}{\end{array}\end{equation*}}

\jmlrheading{1}{2021}{1-48}{4/00}{10/00}{meila00a}{Fan Yao, Chuanhao Li, Denis Nekipelov, Hongning Wang and Haifeng Xu}

\ShortHeadings{User Welfare Optimization with Competing Content Creators}{User Welfare Optimization with Competing Content Creators}
\firstpageno{1}

\begin{document}

\title{User Welfare Optimization in Recommender Systems with Competing Content Creators}

\author{\name Fan Yao$^1$ \email fy4bc@virginia.edu
       \AND
       \name Yiming Liao $^2$ \email yimingliao@meta.com
       \AND
       \name Mingzhe Wu$^1$ \email mw3nzd@virginia.edu 
       \AND
       \name Chuanhao Li$^4$ \email chuanhao.li.cl2637@yale.edu 
       \AND
       \name Yan Zhu$^5$ \email yanzhuyz@google.com
       \AND
       \name James Yang$^2$ \email jamesjy@meta.com
       \AND
       \name Qifan Wang$^2$ \email wqfcr@meta.com
       \AND
       \name Haifeng Xu$^3$ \email haifengxu@uchicago.edu
       \AND
       \name Hongning Wang$^1$ \email wang.hongn@gmail.com 
       \AND \\
       \addr $^1$Department of Computer Science, University of Virginia, USA \\
       \addr $^2$Meta, USA\\
       \addr $^3$Department of Computer Science, University of Chicago, USA\\
       \addr $^4$Department of Statistics and Data Science, Yale University, USA\\
    \addr $^5$Google, USA
    }
\maketitle

\begin{abstract} 
Driven by the new economic opportunities created by the creator economy, an increasing number of content creators rely on and compete for revenue generated from online content recommendation platforms. This burgeoning competition reshapes the dynamics of content distribution and profoundly impacts long-term user welfare on the platform. 
However, the absence of a comprehensive picture of global user preference distribution often traps the competition, especially the creators, in states that yield sub-optimal user welfare. To encourage creators to best serve a broad user population with relevant content, it becomes the platform's responsibility to leverage its information advantage regarding user preference distribution to accurately signal creators. In this study, we perform system-side user welfare optimization under a competitive game setting among content creators. We propose an algorithmic solution for the platform, which  dynamically computes a sequence of weights for each user based on their satisfaction of the recommended content. These weights are then utilized to design mechanisms that adjust the recommendation policy or the post-recommendation rewards, thereby influencing creators' content production strategies. 
To validate the effectiveness of our proposed method, we report our findings from a series of experiments, including: 1. a proof-of-concept negative example illustrating how creators' strategies converge towards sub-optimal states without platform intervention; 2. offline experiments employing our proposed intervention mechanisms on diverse datasets; and 3. results from a three-week online experiment conducted on a leading short-video recommendation platform. 
\end{abstract}

\section{Introduction}

Online content recommendation platforms have evolved into an indispensable component of our daily lives \citep{bobadilla2013recommender}. These platforms play a pivotal role in assisting their users in navigating the vast ocean of content generated by revenue-seeking creators, including various social media platforms (e.g., Facebook, Instagram), streaming services (e.g., YouTube, TikTok), and many more. One of the primary functions of these recommendation platforms is to advance user welfare, defined as the overall volume and quality of interactions between users and content. This metric is widely regarded as a fundamental indicator of the well-being of an online ecosystem and is also closely tied to the platform's revenue. 

After decades of effort in relevance-driven matching between users and content, industry practitioners and researchers have reached the consensus that user welfare optimization cannot be achieved through myopic approaches that merely target at eliciting and predicting user preferences   \citep{qian2022digital,boutilier2023modeling,zhan2021towards,mladenov2020optimizing,yao2022learning,yao2022learning2,dean2022preference,biyik2023preference}. One primary reason is because any matching strategy has a profound impact on content creators' beliefs about the users' demand and consequently their reactions, i.e., what to produce next, leading to a shift in the distribution of content available for recommendation. 
This influence pathway is unfortunately overlooked in existing recommendation algorithm design; and therefore, there is a great need for a robust recommendation strategy that operates with respect to creators' strategic responses and the resultant content dynamics. It is imperative for the platform to encourage creators in generating content that continuously contributes to the overall health of the ecosystem. 

Typically, creators' well-being is intricately linked to the exposure of their content and the economic incentives they accrue from the platform, compelling them to continuously strive for maximized benefits \citep{glotfelter2019algorithmic,hodgson2021spotify}. This dynamic creates a competitive environment that leads to intriguing phenomena in terms of welfare guarantees at equilibrium \citep{fleder2009blockbuster,jagadeesan2022supply,zhu2023online}. 
For instance, \citet{yao2023bad} introduced a game theoretical framework to investigate competition dynamics among content creators. Their research revealed that social welfare loss can be attributed to factors such as the degree of exploration in users' decision making and the span of recommendation slots. As indicated by many previous studies, the platform suffers from sub-optimal social welfare and thus undermines long-term revenue when content distribution lacks necessary diversity to cater to various users' preferences. This issue is also observed in empirical studies, where content creators often exhibit a tendency to chase trends \citep{holmbom2015youtuber,nandagiri2018impact}. 
In essence, creators tend to produce content that arouses the interests of the majority user group, owing to the group's high visibility and the creators' myopic creation strategies \citep{yao2023bad,jagadeesan2022supply}. 
However, it is our contention that the platform should not simply blame creators for their perceived selfishness and myopia. This is because creators do not possess a holistic view of the demand distribution, i.e., user preferences. Instead, it is the platform's responsibility to \textit{disseminate} knowledge about user demand to creators. By doing so, creators can make better informed decisions that mutually benefit their own interest and enhance user welfare (and hence platform's revenue).

In this study, we extend the Content Creator Competition ($C^3$) framework introduced by \cite{yao2023bad,yao2023rethinking}, to model the dynamics of competition among content creators. We relax the behavioral assumptions about creators' updating strategies in the original framework and explore how the platform can design mechanisms to optimize user welfare accordingly. Our key idea is to direct creators' attention towards currently under-served users, by manipulating creators' received utilities with respect to the cumulative user satisfaction about the recommended content. We present a series of approaches to implement the interventions with theoretical justifications.

To validate the effectiveness of our approach, we conducted offline experiments using both synthetic data and the MovieLens dataset, and demonstrated how our mechanism improves user welfare over time under a creator response simulator. Additionally, we deployed an online experiment on a leading short-video recommendation platform over a span of three weeks and observed statistically significant and positive result in terms of the overall user engagement and content diversity. Our model and online experiments offer valuable insights into the design of incentive-aware recommender platforms. To summarize, our contributions can be listed as follows:

\begin{enumerate}[leftmargin=5mm]
    \item We formalize the user welfare optimization problem in a competitive content creation environment and identify the primary cause for potential sub-optimal outcomes: the information asymmetry between content creators and the platform.
    \item We propose a dynamic user importance reweighting approach with theoretical justifications for optimizing user welfare and three implementation schemes which can be applied to various practical scenarios. 
    \item We demonstrate the effectiveness of our solution with both offline simulations and online testing on real traffic. 
\end{enumerate}

\section{Related Work}
\label{sec-related}
The characterization and optimization of long-term dynamics on content platform involving strategic content creators has garnered increasing attention from both theoretical \citep{ben2017shapley,ben2018game,yao2023rethinking,yao2023bad,zhu2023online,hu2023incentivizing,immorlica2024clickbait,hron2022modeling,jagadeesan2022supply,dean2024recommender,immorlica2024clickbait,xu2024ppa,yao2024human} and empirical \citep{mladenov2020optimizing,prasad2023content} fields. Seminal works from \citet{ben2017shapley,ben2018game} introduced a game theoretical setting to model interactions between content creators and users, and proposed the Shapley mediator to ensure the existence of a pure Nash Equilibrium \citep{nash1950equilibrium}.


Recently, \citet{yao2023bad} demonstrated that due to creators' competition, the user welfare loss under a top-$K$ recommender systems can be upper-bounded by $O(\frac{1}{\log K})$. This finding suggests that the platform can improve user welfare by providing more recommendations. Building on this, the authors further proposed a category of mechanisms for the platform to ensure a stable equilibrium and developed a computational solution to identify the optimal mechanism for social welfare optimization \citep{yao2023rethinking}. Additionally, \citet{zhu2023online} introduced an online learning method to jointly optimize recommendation policy and payment contracts for creators to maximize accumulated utility. \citet{hu2023incentivizing} designed a learning algorithm to incentivize the creation of high-quality content. However, all these studies rely on strong behavioral assumptions about content creators, e.g., they can perform no-regret learning \citep{yao2023bad}, or have oracle access to their utility functions \citep{yao2023rethinking,ben2017shapley,ben2018game}, so that the Nash equilibium is achievable. Our work bridges this gap by developing a system-side solution to optimize user welfare that even when creators are not able to achieve Nash equilibria.


On the empirical side, \citet{mladenov2020optimizing} explored a scenario where content creators may leave the platform if their user engagement falls below a threshold. The study optimized social welfare by solving a constrained matching problem. In a similar spirit, \citet{prasad2023content} introduced a sequential prompting policy aimed at optimizing user welfare in equilibrium. The optimal policy was determined through mixed integer programming. The solutions were reported to be effective under specific behavioral assumptions or environmental contexts, e.g., the platform can send prompts to creators as additional signals. However, the platforms are often constrained in their ability to influence the ecosystem. They may primarily rely on monetary incentives to motivate creators and have limited flexibility to manipulate factors beyond matching strategies and post-matching rewards. Our solution addresses this broader range of scenarios, making it applicable, for example, when creators are highly responsive to monetary incentives, and the platform's influence is primarily exerted through adjustments to matching probabilities and post-matching rewards.

\section{The Modeling of Content Creation Competition}

In this section, we formulate the competition among content creators (i.e., players) as a strategic game, which will serve as an environment for the subsequent mechanism design problem. At a high level, each creator's utility is determined by the platform's matching strategy and the post-matching reward function. Creators adhere to simple, local update principles to sequentially alter their strategies, resulting in a dynamic content distribution on the platform. 
The primary objective of the platform is to optimize the cumulative user welfare by designing its matching strategy and post-matching reward function. Our strategic game setup builds upon and extends the framework of Content Creator Competition ($C^3$) game introduced by \citet{yao2023bad,yao2023rethinking}. For the sake of simplicity in nomenclature, we retain the name of $C^3$ and refer to our game as \cgame{}, i.e., an extension of the $C^3$. Formally, a \cgame{} instance is defined by the following tuple: $\big(\X,\{\S_i\}_{i=1}^n, \sigma, \beta, K, R(\cdot)\big)$, which we explain in details below.

\begin{enumerate}[leftmargin=15pt]
    \item \textbf{Basic setups: } a user distribution $\X$ with finite support $\{\x_j\in \RR^d\}_{j=1}^m$, and a set of content creators denoted by $[n]=\{1,\cdots,n\}$. Each creator $i$ can take an action $\s_i$, is often referred to as a \emph{pure strategy} in game-theoretic literature, from an action set $\S_i \subset \RR^d$. $\s_i$ can be understood as the embedding of content that creator $i$ will produce. 
    Without loss of generality, we assume the $L_2$ norms of any $\x$ and $\s_i$ are upper bounded by 1. 
    \item \textbf{Relevance function:} the relevance function $\sigma(\s, \x): \RR^d\times \RR^d \rightarrow \RR_{\geq 0}$ measures the \emph{relevance} score between a user $\x \sim \X$ and content $\s$. Without loss of generality, we normalize $\sigma$ to $[0, 1]$, where $1$ suggests perfect matching. We focus on modeling the strategic behavior of creators and thus abstract away the estimation of $\sigma$ \footnote{We assume $\sigma$ is learned from the offline data and $\sigma(\s_i, \x)$ is an unbiased estimation of user $\x$'s satisfaction when exposed to $\s_i$. }. For simplicity, we use $\sigma_{i,\x}$ to denote $\sigma(\s_i, \x)$ when the joint strategy profile $\s=(\s_1,\cdots,\s_n)\in\S$ and user profile $\x$ are clear in the context of our discussion. 
    \item \textbf{Matching function:} Given any user $\x\in \X$ and when each creator commits to a strategy $\s_i$, the platform retrieves the top-$K$ ranked content in terms of the relevance scores $\{\sigma_{i,\x}\}_{i=1}^n$ and match one of them to $\x$. Specifically, let $\{\sigma_{l(1),\x}\geq \cdots\geq \sigma_{l(n),\x}\}$ be a permutation of $\{\sigma_{i,\x}\}_{i=1}^n$, we assume that the platform would pick $\s_{\x}\in L_{\x}(K;\s)\triangleq \{\sigma_{l(i),\x}\}_{i=1}^K$ using a softmax distribution with temperature $\beta\geq 0$ \footnote{The formulation in \citep{yao2023bad} also assumes the platform retrieve top-$K$ content for each user, but let the user to choose one according to the Random Utility model. The resulting matching probability shares the same form as in Eq. \eqref{eq:matching_rule}, but differs in the sense that the $\beta$ in our setting is a parameter controlled by the platform while it is the user decision noise in \citep{yao2023bad}.}, i.e., 
    \begin{equation}\label{eq:matching_rule}
        P_i(\s,\x) \triangleq Prob [\s_{\x}=\s_{l(i)}] \propto \exp[\beta^{-1} \sigma_{l(i),\x}], 1\leq i\leq K. 
    \end{equation}
    A small $\beta$ makes the matching strategy more deterministic, and $\beta \rightarrow \infty$ corresponds to random matching.
    \item \textbf{User utility and welfare: } When user $\x$ is matched with $\s$, the user's perceived utility is given by a function $\pi(\s, \x)$. The user welfare $W(\s)$ is thus defined as the total expected utility resulted from the matching,
    \begin{equation}\label{eq:welfare_func_general}
        W(\s) = \EE_{\x\sim\X}[\pi(\s_{\x},\x)].
    \end{equation}
    To simplify the technical discussions, we assume the learned relevance function $\sigma$ is an unbiased estimation of $\pi$, and therefore $W(\s)$ can be simplified to   
    \begin{equation}\label{eq:welfare_func}
        W(\s) = \EE_{\x\sim\X}[\sigma(\s_{\x},\x)].
    \end{equation}
    However, our proposed solution works for general welfare function defined in Eq \eqref{eq:welfare_func_general}.
    \item \textbf{Creator utility: }
    For creator $i$, her utility is given by
    \begin{equation}\label{eq:creator_utility_func}
            u_i(\s)=\EE_{\x\in\X} \left[R(\s_i, \x) \cdot P_i(\s,\x)\right],
    \end{equation}
    where $R(\s_i, \x)$ is the system-provided reward for this matching.
    
    Natural choices of $R$ include $R(\s_i, \x)$ being proportional to the user's perceived utility, or simply setting $R(\s_i, \x)=1$ (i.e., reward creators by the amount of traffic). Therefore, we have
    \begin{align}\label{eq:creator_expected_utility_func}
            &u_i(\s)=\EE_{\x\in\X} \left[\sigma(\s_i, \x) \cdot P_i(\s,\x)\right], \\ 
            \label{eq:creator_expected_utility_func2}
            &u_i(\s)=\EE_{\x\in\X} \left[P_i(\s,\x)\right],
    \end{align}
    Throughout the paper we adopt Eq \eqref{eq:creator_expected_utility_func} as the platform's default choice, as it is demonstrated in \citep{yao2023bad} that rewarding creators by user utility enjoys a better welfare guarantee than rewarding them by traffic.
  \end{enumerate}

The most well established concept for characterizing a game's outcome is pure Nash equilibrium (PNE) \citep{nash1950equilibrium}. At a PNE, any possible deviation from a player's current strategy would not increase her utility conditioned on other players' strategies. Under some mild assumptions, we can prove that the PNE of our \cgame{} game exists and is unique as stated in the following theorem. 

\begin{theorem}\label{thm:PNE}
 Any \cgame{} game with $K=n$ has a unique pure Nash equilibrium (PNE) under the utility function \eqref{eq:creator_expected_utility_func2} if $\sigma(\cdot)$ is sufficiently smooth and concave and each creator has a convex strategy set. 

\end{theorem}

Theorem \ref{thm:PNE} guarantees the existence of a unique PNE and thus theoretically allows the platform to establish a stable outcome. However, in practical scenarios, we find it uninteresting to either generalize this result or delve further into its properties for two reasons. First, it is rare for $K=n$ to hold in practice because no system will present the entire collection of content to each user.
When $K<n$, the existence of a PNE becomes challenging to establish, due to the discontinuity of the utility functions caused by the top-$K$ ranking operator during the matching process. 
Second, even when a PNE does exist, it does not suggest that creators can consistently reach it through sequential updates. Furthermore, the existence of a PNE does not necessarily imply it is easily achievable in practice, nor does it suggest an improved user welfare. In fact, as we will demonstrate in Section \ref{sec:failure}, even in a simple environment with a unique PNE, a natural updating dynamics among creators fails to converge to the PNE and results in sub-optimal user welfare.

Therefore, we focus on a more practical solution concept called \emph{Local Nash equilibria (LNE)}. While a PNE requires that all players do not want to deviate to any other strategy in the entire space, an LNE merely stipulates players are satisfied with their strategies in a local region. Its formal definition is given as follows. 

\begin{definition}\label{def:LNE} \emph{A profile of creator strategies $\{ \s^*_i \}_{i=1}^n$ forms a local Nash equilibrium (LNE), if for every creator $i$, there exists an open set $\S^0_i\in\S_i$ such that $\s^*_i$ is a best response strategy within $\S_i^0$; formally, 
 \begin{equation}\label{eq:lne-def}
      u_i(\s^*_i,\s^*_{-i}) \geq   u_i(\s_i,\s^*_{-i}) \, \,  \text{ for every } \s_i \in \S^0_i.
 \end{equation}}
 \end{definition} 

We argue that LNE offers a more intuitive and practical solution concept for consideration due to two observations. First, the strategic evolution of content creation is often deeply intertwined with creators' historical decisions \citep{kajander2019challenges}. This correlation stems from content generation being anchored in domain-specific expertise and accumulated experiences, which are inherently stable attributes. As a result, the produced content usually demonstrates path dependency, posing significant challenges for creators in implementing drastic modifications. Second, creators are typically constrained by a lack of comprehensive insights into their utility functions due to a limited understanding of the user demographic and the distribution of user preferences. Given these constraints, creators are likely to resort to incremental adjustments for strategy update.


Hence, we focus on the setting where creators engage in a repeated play of \cgame{} and employ a local searching rule termed local better response (LBR) update for improving their strategies. The details of LBR is presented in Algorithm \ref{alg:creator_dynamic} in Appendix \ref{app:lbr}. LBR characterizes two fundamental properties of content creation: 1. it relies solely on point estimations of the utility function; and 2. it only incurs local changes at each update. At each step, a creator who decides to update her strategy would first generate an exploration direction $\g_i$ and then she would evaluate whether adjusting her strategy in this direction results in a higher utility. If so, she proceeds to update her strategy along $\g_i$ in a pace of $\eta$; otherwise, she maintains her current strategy. This procedure closely emulates real-world scenarios where creators strive to optimize their utilities while having merely black-box access to the utility functions. In practice, finding a clear direction that guarantees improved utility can be a challenging and, at times, unrealistic task. Consequently, we model their strategy evolution as an iterative process of trial and error. By definition, when LBR converges in \cgame{}, it must converge to an LNE. Our primary interest lies in understanding how the platform can devise a dynamic rewarding or matching principle that maximizes cumulative user welfare within a given time period.




\section{Intervention Mechanism Design}
In this section, we introduce the new intervention mechanism designed to optimize user welfare. These mechanisms are intended for the platform to influence creators' perceived utilities, thereby guiding the evolution of their strategies toward more desirable outcomes. We will first establish the need for platform-driven mechanism design by illustrating how suboptimal results can arise in a simplified example without any intervention. Subsequently, we will delve into the specifics of our proposed methods.

\begin{figure}[t]
  \centering
  \includegraphics[trim=40 40 40 40, clip,width=0.49\linewidth]{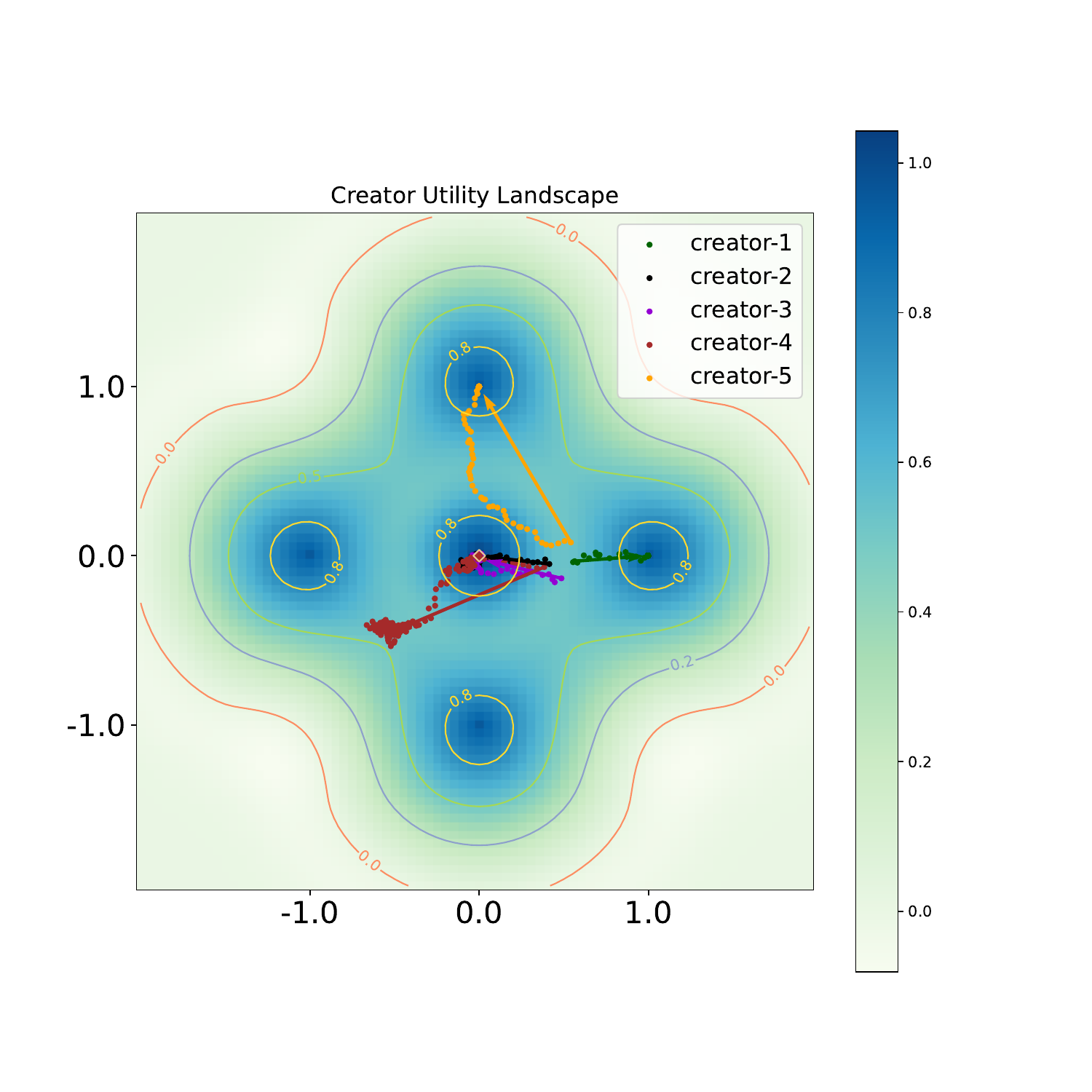}
  \includegraphics[trim=40 40 40 40, clip,width=0.49\linewidth]
{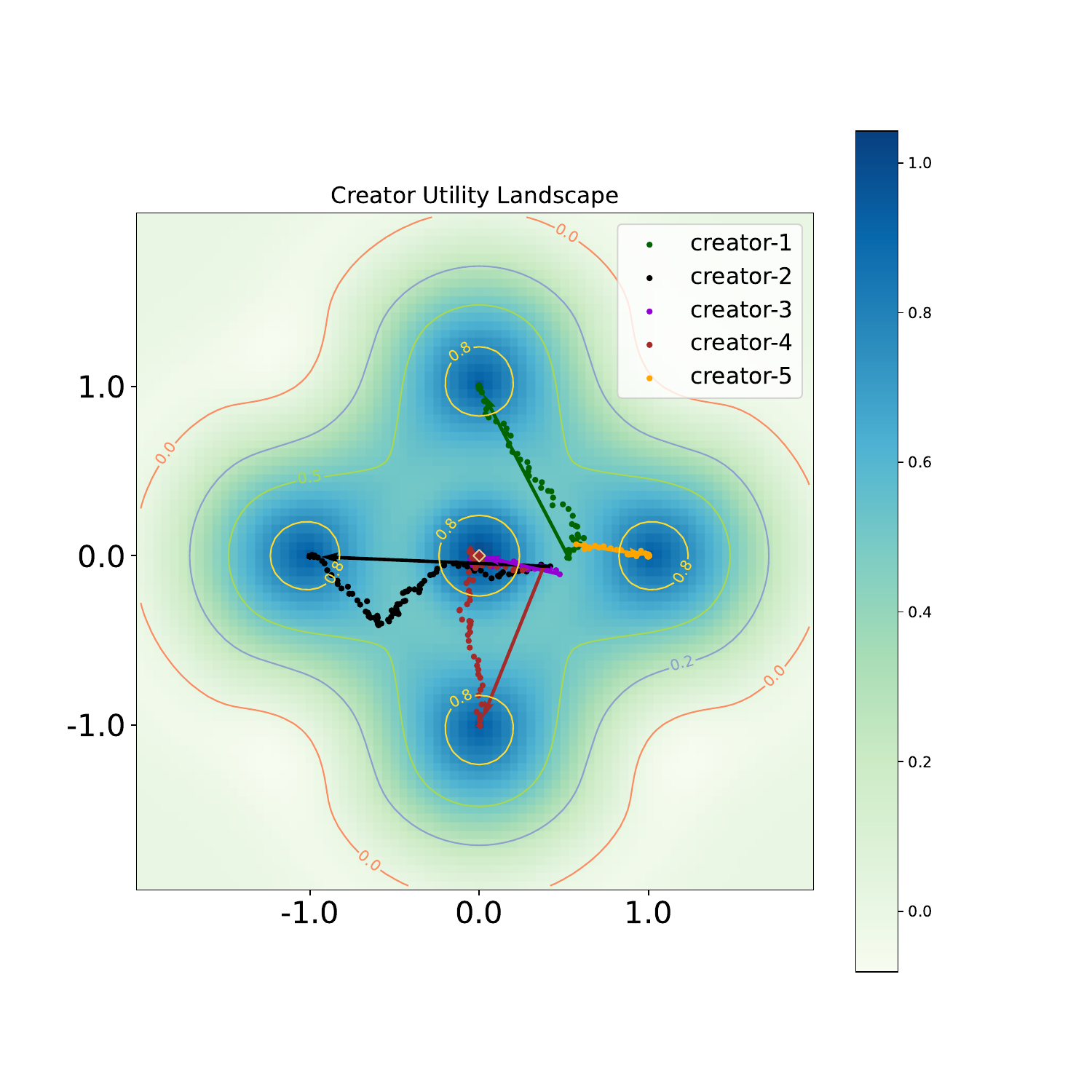}
\vspace{-3mm}
  \caption{Visualization of creators' evolving strategies. Left: no intervention, right: platform decreases the weight of the center user by half. Creators' strategies are marked with different colors, and the arrows start from initial strategies and point to the last-iterate strategies.}
  \label{fig:example}
\vspace{-3mm}  
\end{figure}

\subsection{The Necessity of Intervention}\label{sec:failure}


We start with a simple illustrative example to show how the competition among creators could result in quite inferior user welfare in \cgame{} when creators employ local update dynamics specified in Algorithm \ref{alg:creator_dynamic}. This example exhibits a stark contrast to the sound welfare guarantee for no-regret learning \citep{belmega2018online} equipped creators in \citep{yao2023bad}. Consider a \cgame{} instance $(\X,\{\S_i\}_{i=1}^n, \sigma, \beta, K, R(\cdot))$ described below. The user population $\X$ is evenly distributed over the finite set $ \{\x_j\}_{j=1}^5= \{(0,0),(1,0),(0,1),(-1,0),(0,-1)\}$ and there are $n=5$ content creators, each with action set $\S_i=\RR^2$. The reward function is defined as $R(\s_i,\x)=\sigma(\s_i,\x)=\max\{2-\|\s_i-\x\|_2, 0\}$ and $\beta=10, K=3$. It is evident that the user welfare defined in Eq \eqref{eq:welfare_func} is maximized when each creator precisely targets a single user, i.e., $\s_i=\x_i, 1\leq i\leq 5$, which also represents the PNE of this game. However, as we will illustrate through simulations, creators' strategies do not converge to the PNE nor optimize the user welfare under the LBR dynamics when the platform does not intervene. 

First, let's examine what happens when the platform takes no action to guide the creators. The left panel of Figure \ref{fig:example} visualizes the trajectories of strategy evolution in our constructed environment. Initially, creators' strategies are randomly distributed in the region between $\x_1$ and $\x_2$. Over time, $\x_2$ and $\x_3$ are exclusively occupied by one creator each, while $\x_1$ has two creators competing for it. The remaining creator chooses not to target either $\x_4$ or $\x_5$ and hovers around the region between $\x_4$ and $\x_5$, leaving both $\x_4$ and $\x_5$ unsatisfied. 

From the observed strategy evolution paths, we can deduce how this sub-optimal situation arises. Initially, creators move in different directions: two creators quickly converge to $\x_2$ and $\x_3$, while the remaining three compete for the attention of the central user $\x_1$. However, after this point, no creator has a strong incentive to move closer to $\x_4$ or $\x_5$, as the marginal utility gained from getting closer to $\x_4$ or $\x_5$ does not compensate for the loss incurred by moving away from $\x_1$. Consequently, two creators decide to remain around $\x_1$ and one creator settles in a region between $\x_4$ and $\x_5$.

The above observations highlight the pivotal role played by the central user $\x_1$ in the occurrence of sub-optimal results. Since $\x_1$ is close to other users in the embedding space, targeting $\x_1$ becomes a popular and safe choice for creators. It secures a fraction of attention from $\x_1$ without completely sacrificing the utility gained from other user groups. Thus, users like $\x_1$ act as ``popular states'' when creators dynamically adjust their strategies. Whenever a creator is located near $\x_1$, they are likely to be trapped and reluctant to explore potentially better strategies. Consequently, such ``popular'' users end up attracting more creators, leaving other users unattended.

One immediate solution for the platform is to identify and reduce the impact of these ``popular" users. For instance, the platform can halve the utility gained from the central user $\x_1$ for each creator. This simple mechanism works effectively in this example, as illustrated in the right panel of Figure \ref{fig:example}. Initially, there are still three creators converging to $\x_1$. However, due to the reduced reward from $\x_1$, two creators find it less profitable to stay, driving them to deviate towards $\x_4$ and $\x_5$. By assigning different importance weight for each user, the platform can reshape each creator utility landscape and therefore influence their local search based dynamical behaviors.

\subsection{Platform's Intervention Mechanisms}

The observations above motivate our design of intervention mechanisms that can be employed by the platform to influence creators' perceived utilities. These mechanisms lay the foundation for the adaptive optimization methods we will delve into later. As a reminder, as defined in Eq \eqref{eq:creator_utility_func}, a creator's expected utility from a specific user $\x$ is influenced by two key factors: the probability of creator $i$ being matched with user $\x$ denoted as $P_i(\s,\x)$, and the post-matching reward assigned by the platform, denoted as $R(\s_i, \x)$. 
The default choice of the platform is to set the reward function $R(\s_i, \x)=\sigma(\s_i, \x)$ as in Eq \eqref{eq:creator_expected_utility_func} and the matching probability function $P_i(\s,\x)$ as the softmax over the top-$K$ ranked content $\s_i$ as demonstrated in Eq \eqref{eq:matching_rule}.


In the example provided in Section \ref{sec:failure}, the primary factors leading to sub-optimal welfare is the presence of popular user groups that attract excessive creator attention, making minority user groups unnoticed by creators. To enhance overall user welfare, it is crucial for the platform to guide creators' attention toward these overlooked user groups by re-emphasizing their significance. In this way, creators who were previously unaware of these user groups or found them less lucrative may consider adjusting their strategies to align more closely with those users' preferences. To achieve this objective, we introduce and study three different approaches for modifying the schemes of $R(\s_i, \x)$ and $P_i(\s,\x)$, namely User Importance Reweighting (UIR), Soft Matching Truncation (SMT), and Hard Matching Truncation (HMT). These three mechanisms share a common underlying principle, but they are designed to operate under different scenarios, taking into account potential constraints faced by a platform.

\noindent {\bf User Importance Reweighting (UIR)} 

\noindent
The most straightforward approach is UIR,
\begin{equation}\label{eq:augmented_utility_UIR}
    u_i(\s_i,\s_{-i})=\EE_{\x\in \X} [w(\x)\cdot R(\s_i, \x)\cdot P_i(\s,\x)],
\end{equation}
where the platform simply adjusts the post-matching rewards for creators based on the measured importance of each user. Specifically, if the platform believes a user has been under-served under the current content distribution, it raises the reward for creators whose content is consumed by such a user. Intuitively, this sends a message to creators that ``if you shift your content towards such users, you will get a higher marginal reward compared to sticking to your current content.'' As a result, the platform can carefully design the user weights such that a reasonable number of creators can be successfully incentivized to serve the targeted users. 

\noindent \noindent {\bf Soft Matching Truncation (SMT) and Hard Matching Truncation (HMT)} 

\noindent
Both SMT and HMT function in a similar manner as UIR but focus on manipulating the matching probability rather than the post-matching reward by utilizing the weight $w(\x)$. Recall that the probabilistic matching function $P$ is characterized by two parameters: the truncation number $K$ (which, in practice, corresponds to the total number of recommendation candidates retrieved for ranking) and the temperature $\beta$ (which can be viewed as a measure of the exploration strength in the ranking model). 
When the platform needs to signal the importance of a specific user $\x$, it enhances $\x$'s visibility among creators, increasing the chance that creators who were previously unaware of $\x$ start realizing the potential benefits of catering to $\x$. This can be achieved by either increasing $\beta$ or $K$: increasing $\beta$ flattens the distribution of $\x$'s matches among the top-$K$ candidates, while increasing $K$ enlarges the pool of creators exposed to $\x$. Therefore, both of them augment the expected number of creators exposed to $\x$. Since $K$ imposes a rigid threshold on the number of creators exposed to $\x$, while $\beta$ offers a more flexible threshold, we refer to them as Hard Matching Truncation (HMT) and Soft Matching Truncation (SMT), respectively:
\begin{align}\label{eq:augmented_utility_SMT}
    u_i(\s_i,\s_{-i})=\EE_{\x\in \X} [R(\s_i, \x)\cdot P_i(\s,\x; \beta(w(\x)),K)],\\
\label{eq:augmented_utility_HMT}
    u_i(\s_i,\s_{-i})=\EE_{\x\in \X} [R(\s_i, \x)\cdot P_i(\s,\x; \beta,K(w(\x)))].
\end{align}


We remark that UIR is more suitable when the platform possesses the flexibility to design payment incentives for creators. However, if the platform has limited control over payment, such as budget constraints or other factors, SMT or HMT can be employed, as they only require minor adjustments to the matching function. The specific choices of increasing functions $\beta(\cdot),K(\cdot)$ are flexible and we leave it to the experiments.

\subsection{Welfare Optimization through Adaptive Reweighing}

To implement our proposed intervention mechanisms, we need to compute the corresponding user-specific weighting functions, namely $w(\cdot)$, $\beta(\cdot)$, and $K(\cdot)$. In this section we will use UIR as an example to illustrate our method and let the user distribution $\X$ be a uniform distribution over its support $\{\x_1,\cdots,\x_m\}$ so that $w(\cdot)$ can be parameterized by a vector $\w\in\RR^m_{\geq 0}$. When the platform commits to an intervention mechanism $\w$, the content creators' strategic updates according to LBR (i.e., algorithm \ref{alg:creator_dynamic}) will lead their joint strategy to an LNE $\s^*$, which determines the content distribution and the total user welfare $W$. Therefore, the task of finding the optimal $\w$  maximizing $W$ under \cgame{} can be formulated as the following bi-level optimization problem:
\begin{align}\label{eq:bi-level-opt}
    \max_{\w\in \RR_{\geq 0}^m} & \quad W(\s^*(\w)) \\ \label{eq:bi-level-opt-inner}
    \text{s.t., } & \quad \s^*(\w) \text{~is an LNE of~\cgame{}}.
 \end{align}

We adopt the formulation in Eq \eqref{eq:bi-level-opt} simply for presentation purpose, as the constraint in Eq \eqref{eq:bi-level-opt-inner} is not well-defined due to the non-uniqueness of LNE of \cgame{} in general. When we takcle problem in Eq \eqref{eq:bi-level-opt}, we employ either LBR for simulating an $\s^*(\w)$ in offline experiments, or we directly observe $\s^*(\w)$ based on the creators' actual responses over a period of time for online experiments. An straightforward approach to solve Eq \eqref{eq:bi-level-opt} is to use an iterative method to dynamically adjust $\w$, and the main challenge is to pin down an improving direction of $\w$. Ideally, we can apply first-order optimization if an estimation of the gradient $\frac{d W}{d \w}$ is available. However, the interplay between $\w$ and $\s^*(\w)$ is generally intractable to analyze and we have to resort to heuristic methods. To get an intuitive idea about an improving direction of $\w$, we consider a stylized setting where the user population is perfectly separated and the relevance function is given by dot-product $\sigma(\s,\x)=\s^{\top}\x$. In such a structured environment, the following theorem reveals a useful principle for finding an improving direction of $\w$.

 \begin{theorem}\label{thm:opt}
  When the number of creators $n$ is large enough and the user population $\X$ is a uniform distribution over an orthogonal basis in $\RR^d$, updating $\w$ with the following formula guarantees an improvement in $W$ defined in Eq \eqref{eq:welfare_func}:
  \begin{equation}\label{eq:84}
    w'_j = w_j \cdot e^{-\eta \bar{\pi}(\x_j)}, \forall j\in [m],
\end{equation}
where $\eta$ is a small scalar denoting the learning rate, and $\bar{\pi}(\x_j)$ is the expected utility of user $\x_j$ at $\s^*(\w)$.
 \end{theorem}

By the definition in Eq \eqref{eq:creator_utility_func}, rescaling each $w_j$ by a constant does not alter the nature of problem in Eq \eqref{eq:bi-level-opt}. Therefore, the insight conveyed by Eq \eqref{eq:84} is clear: if a user enjoys a high expected utility under the current content distribution, the platform should reduce her weight when rewarding creators. Conversely, if a user's expected utility is relatively low, the platform needs to highlight her significance for motivating a larger set of creators to develop content that caters to the needs of this user. Despite the fact that Eq \eqref{eq:84} is derived from a significantly simplified user distribution, we will leverage it as a foundational element in the development of our adaptive reweighing algorithm and demonstrate in our experiments that this simple heuristic works pretty well for real user distributions.

Next, we formally introduce our proposed adaptive reweighting algorithm for optimizing the intervention mechanism $\w$. Each user $\x$ is initially assigned a unit weight $\w^{(0)}(\x)=1$. During subsequent iterations, the platform continuously monitors the average utility of user $\x$, denoted as $\bar{\pi}(\x)$, within a specified time window, and updates $\w$ according to the following \eqref{eq:reweight}, where $\alpha>0$ is a tunable parameter. This adjustment process employs the meta-algorithm structure of multiplicative weight update method \citep{arora2012multiplicative}.
\begin{equation}\label{eq:reweight}
    w^{(i+1)}(\x) \propto w^{(i)}(\x)\cdot \exp(-\alpha \bar{\pi}(\x)).    
\end{equation}

In practice, we can choose the user utility function $\pi(\x)$ as the metric used for defining the user welfare function Eq \eqref{eq:welfare_func_general}. Up to this point, our discussion has primarily focused on the assumption that $\pi(\x;\s)\propto \sigma(\s_i,\x)$. However, it is important to highlight that $\pi$ in Eq \eqref{eq:reweight} can also take alternative forms to optimize empirical performance. For instance, it can be a function of any numerical measurement related to user satisfaction (e.g., click-through rate). To reduce the dimension of the user weight vector and enhance the robustness of weight updates, we recommend that algorithm designers pre-cluster users into $L$ groups based on their static features so that users within the same group maintain identical weights. The platform's intervention strategy is thus parameterized by an $L$-dimensional vector, $\w=(w_1,\cdots,w_L)$, with each entry denoting the weight assigned to the corresponding user group. 

For a fixed time horizon $T$ in which the platform plans to perform intervention, the platform divides the horizon into $E$ epochs, each with an equal length of $M$ (i.e., $T=EM$). At the start of each epoch $e$, the platform commits to a weight vector $\w^{(e)}$ and deploy it to one of the intervention mechanisms UIR, SMT or HMT. After that, the platform observes and records the sequence of creators' strategic responses, denoted as $\{\s^{(e,i)}\}_{i=1}^M$ from the online environment. Subsequently, the algorithm estimates the average user welfare $\bar{\pi}_l$ for each group $l$. 
It then employs values in $\{\bar{\pi}_l\}_{l=1}^L$ to update the weights at the beginning of the $(e+1)$-th epoch using Eq \eqref{eq:reweight}. To prevent $\w$ from growing or declining excessively, after each update we first normalize and then clip its values within a predetermined interval $[w_{\min},w_{\max}]$. The formal description of this process is presented in Algorithm \ref{alg:ada_reweight}. The implementation details for the deployment of UIR, SMT and HMT in Line 4 are deferred to Appendix \ref{app:implement_mechanisms}.

\vspace{-1mm}
\begin{algorithm}[h]
   \caption{Adaptive Reweighting}
   \label{alg:ada_reweight}
\begin{algorithmic}[1]
   \STATE {\bfseries Input:} Number of epochs $E$, Epoch length $M$, Initial strategy profile $\s^{(0)}$, learning rate $\eta$, temperature parameter $\alpha$, user groups $(G_1, \cdots, G_L)$, clipping constant $w_{\min},w_{\max}$.
   \STATE {\bfseries Initialization:} Initial weight $\w^{(0)}=(w_1^{(0)}, \cdots, w_L^{(0)})$.
   \FOR{$e=0$ {\bfseries to} $E$}
        \STATE Deploy the weight $\w^{e}$ using UIR (Eq \eqref{eq:augmented_utility_UIR}), SMT (Eq \eqref{eq:augmented_utility_SMT}) or HMT (Eq \eqref{eq:augmented_utility_HMT}).
        \STATE Observe creators' strategy sequence $\{\s^{(e,i)}\}_{i=1}^M$.
        \STATE Compute the average user utility for each group 
        $$\bar{\pi}_l = \frac{1}{M|G_l|}\sum_{\x \in G_l}\sum_{i=1}^M \pi(\x;\s^{(e,i)}).$$ 
        \STATE Update $w^{(e+\frac{1}{3})}_l = w^{(e)}_l\cdot \exp(-\alpha \bar{\pi}_l), l\in [L]$. 
        \STATE Normalize $w^{(e+\frac{2}{3})}_l = L\cdot w^{(e+\frac{1}{3})}_l / \sum_{j=1}^L w^{(e+\frac{1}{3})}_j, l\in [L]$. 
        \STATE Clip $\w^{(e+1)} = $ {\bf Clip}$(\w^{(e+\frac{2}{3})}, w_{\min},w_{\max})$. 
        \STATE Set $\s^{(e+1)}=\s^{(e,M)}$.
   \ENDFOR
\end{algorithmic}
\end{algorithm}
\vspace{-1mm}

\section{Experiments}

In this section, we evaluate our proposed intervention mechanisms on both offline datasets and an online environment on a leading short-video recommendation platform in the industry. 

\subsection{Experiments on Offline Data}\label{sec:syn_env}

We conduct simulations on \cgame{} game instances constructed from synthetic data and MovieLens-1m dataset \citep{harper2015movielens}. In the following, we first introduce the specification of these two simulation environments and then report the results. 

\subsubsection{Synthetic environment}
For the synthetic environment, we first construct the user population as follows: we fix an embedding dimension $d=5$ and independently sample $10$ cluster centers, denoted as $\{\c_1,\cdots,\c_{10}\}$, from the unit sphere $\mathbb{S}^{d-1}$. For each center $\c_i$, we generate users belonging to cluster-$i$ by independently sampling from a Gaussian distribution $\N(\c_i, 0.5^2 I_d)$. The sizes of the $10$ user clusters are denoted by a vector $\z=10\times(100, 50, 20, 10, 10, 5, 2, 1,1,1)$. In this manner, we generate a population $\X$ of size $m=2000$. 
The number of creators is set to $n=200$, and each action set $\S_i$ is set to the unit ball in $\RR^d$. The user utility and relevance score function are set to $\pi(\s,\x)=\sigma(\s,\x)=\max\{1-\|\s-\x\|/3, 0\}$. We set $(\beta, K)$ to $(0.1, 20)$ by default. Such synthetic datasets characterize a class of clustered user preference distributions (e.g., majority vs., minority user groups).

On the creators' side, we let their initial strategies to be close the center of the largest user group. This environment models a situation where creators tend to chase popular trends by exclusively producing content tailored to the taste of the largest user group. We aim to investigate whether our proposed mechanisms can assist the platform to escape from such sub-optimal states.

\begin{figure*}[t]
\begin{subfigure}{0.31\textwidth}
\includegraphics[trim=0 0 0 0, clip,width=\textwidth]{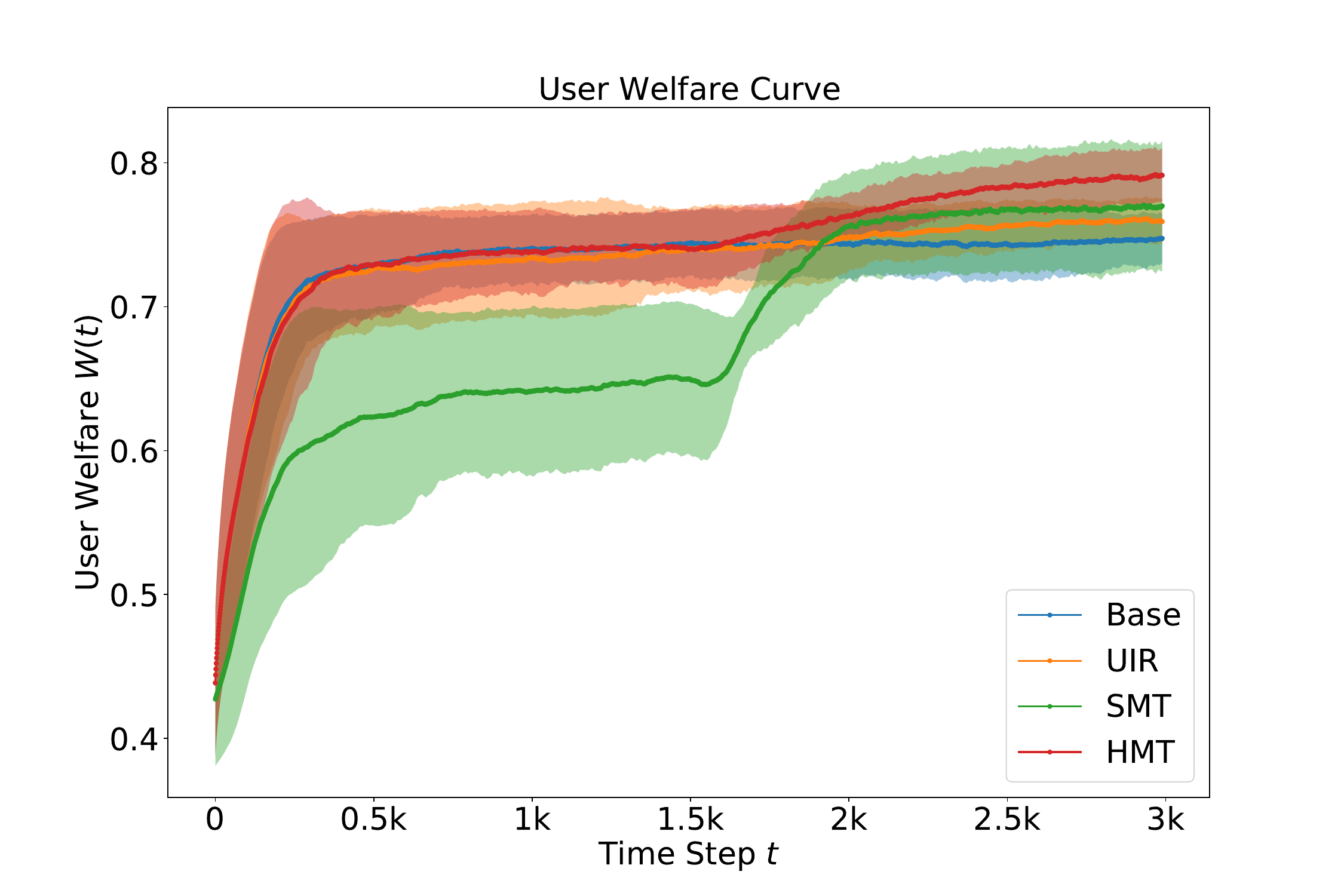}
\caption{User welfare evolving curve.}
\label{fig:welfare_curve_1}
\end{subfigure}
\begin{subfigure}{0.31\textwidth}
\includegraphics[trim=0 0 0 0, clip,width=\textwidth]{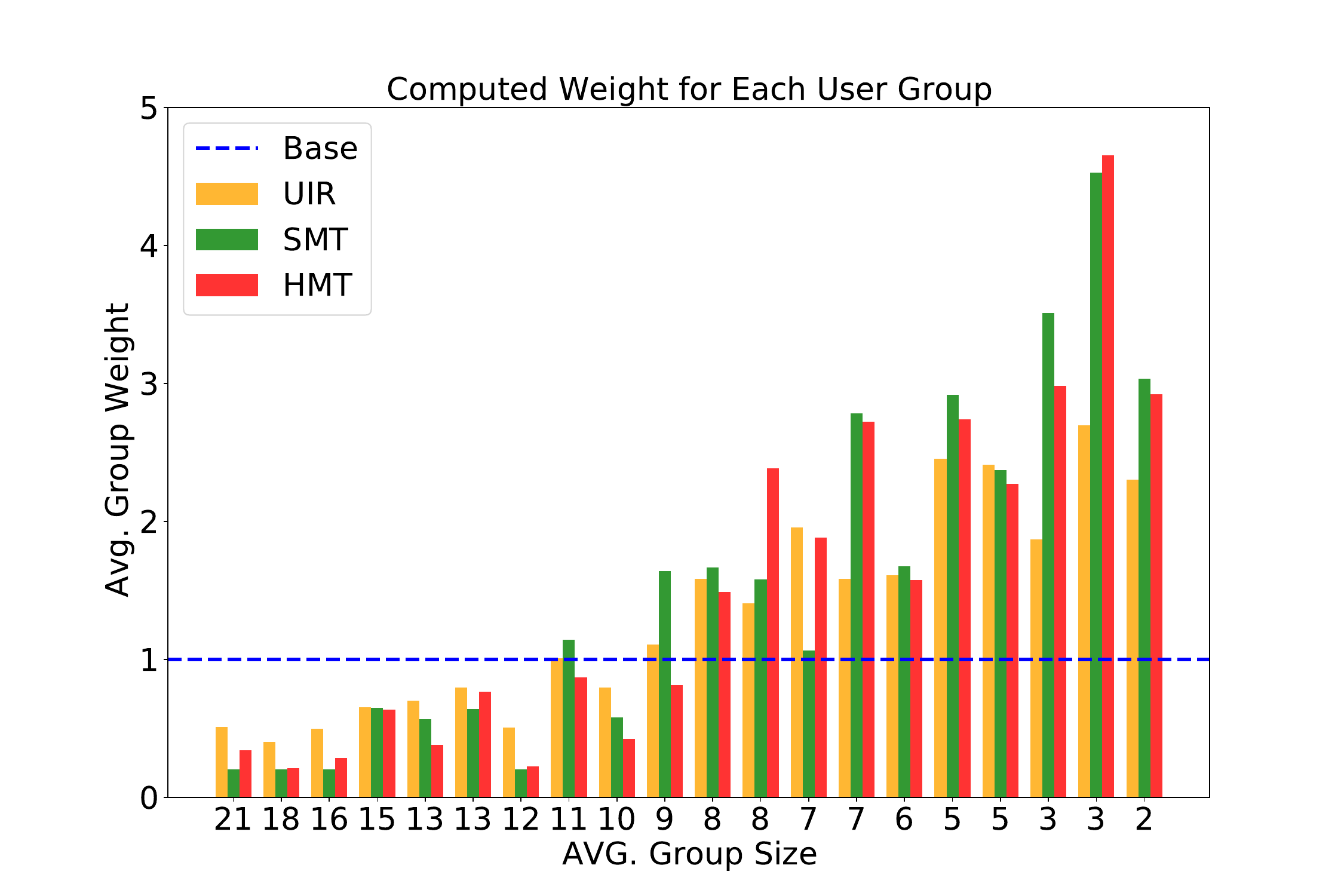}
\caption{Avg. group weights.}
\label{fig:weight_group_1}
\end{subfigure}
\begin{subfigure}{0.31\textwidth}
\includegraphics[trim=0 0 0 0, clip,width=\textwidth]{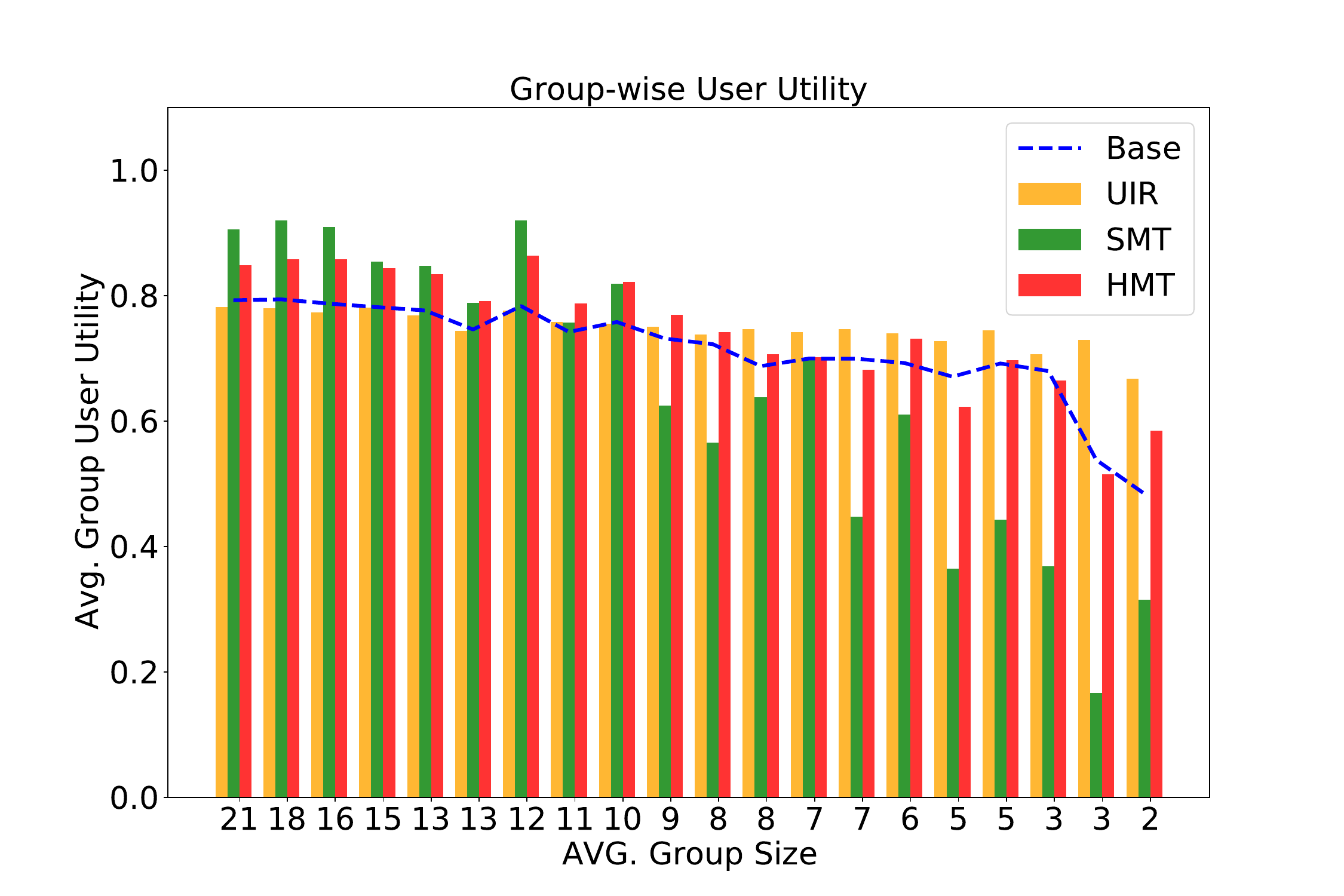}
\caption{Avg. group utilities.}
\label{fig:utility_group_1}
\end{subfigure}
\vspace{-3mm}
\caption{Performance of UIR, SMT and HMT on synthetic dataset against the no-intervention baseline. Results are averaged over 10 independently sampled synthetic environments including one-sigma error bars. $x$-axis: group sizes divided by 10.}
\label{fig:result_1}
\vspace{-3mm}
\end{figure*}

\begin{figure*}[h]
\begin{subfigure}{0.31\textwidth}
\includegraphics[trim=0 0 0 0, clip,width=\textwidth]{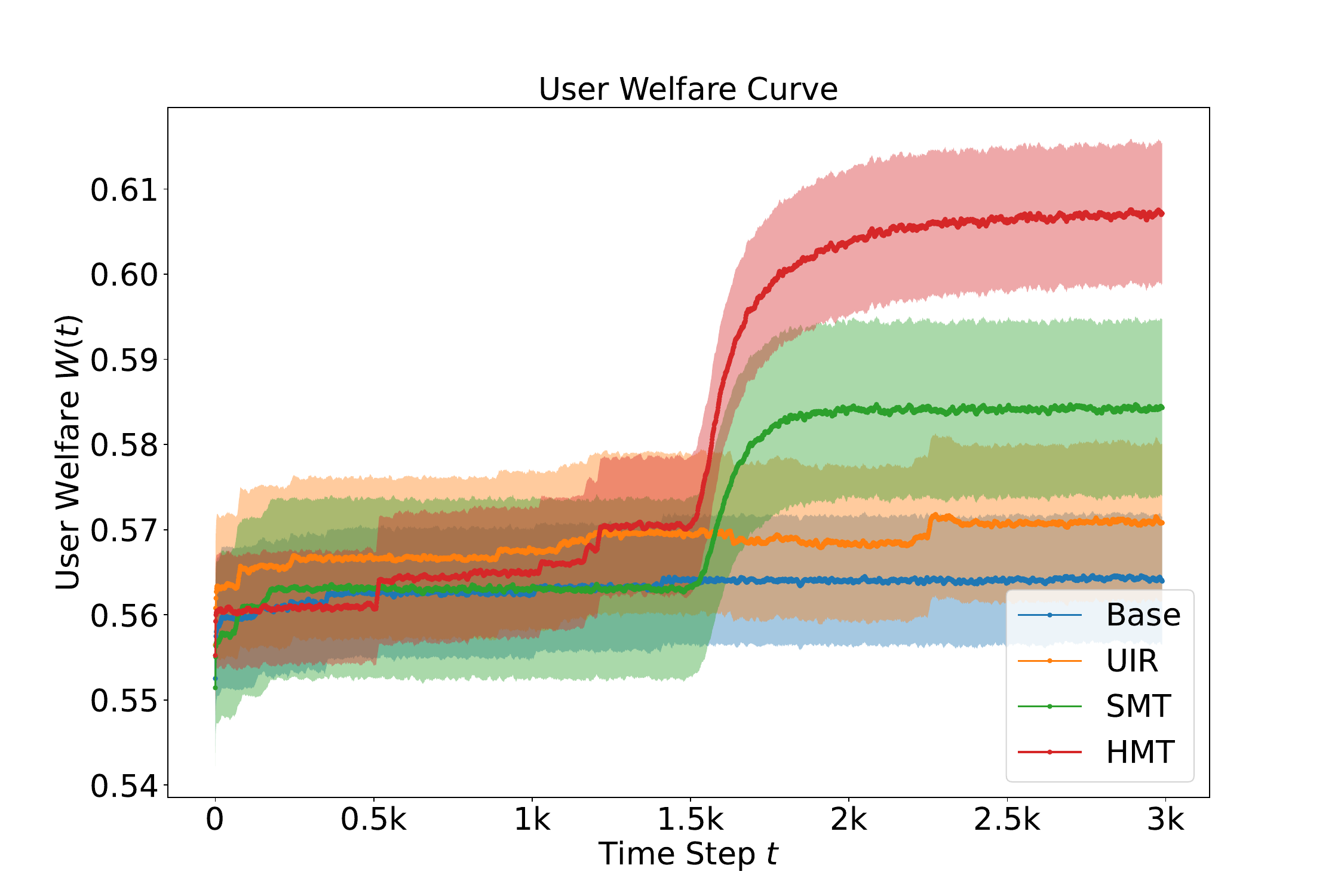}
\caption{User welfare evolving curve.}
\label{fig:welfare_curve_2}
\end{subfigure}
\begin{subfigure}{0.31\textwidth}
\includegraphics[trim=0 0 0 0, clip,width=\textwidth]{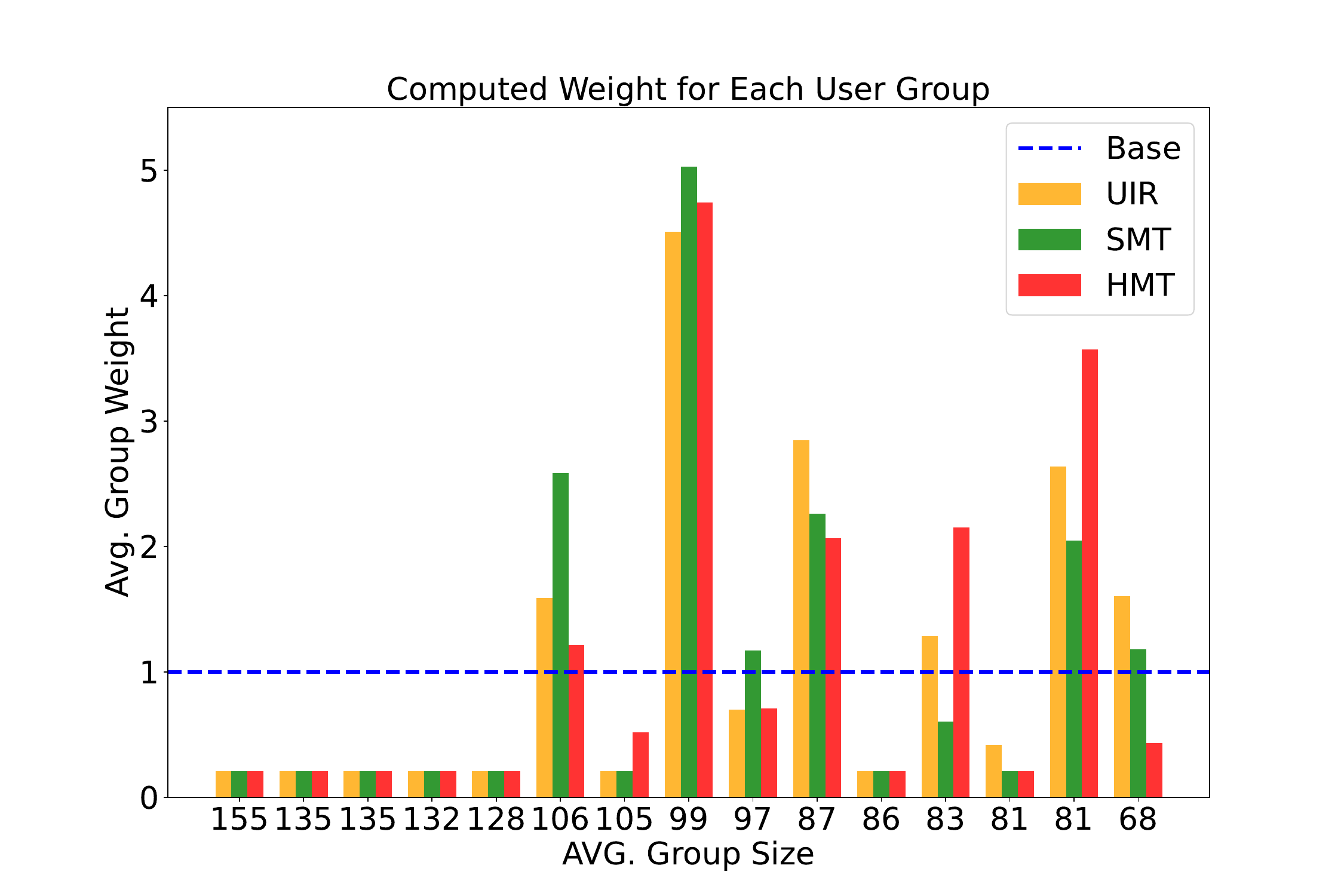}
\caption{Avg. group weights.}
\label{fig:weight_group_2}
\end{subfigure}
\begin{subfigure}{0.31\textwidth}
\includegraphics[trim=0 0 0 0, clip,width=\textwidth]{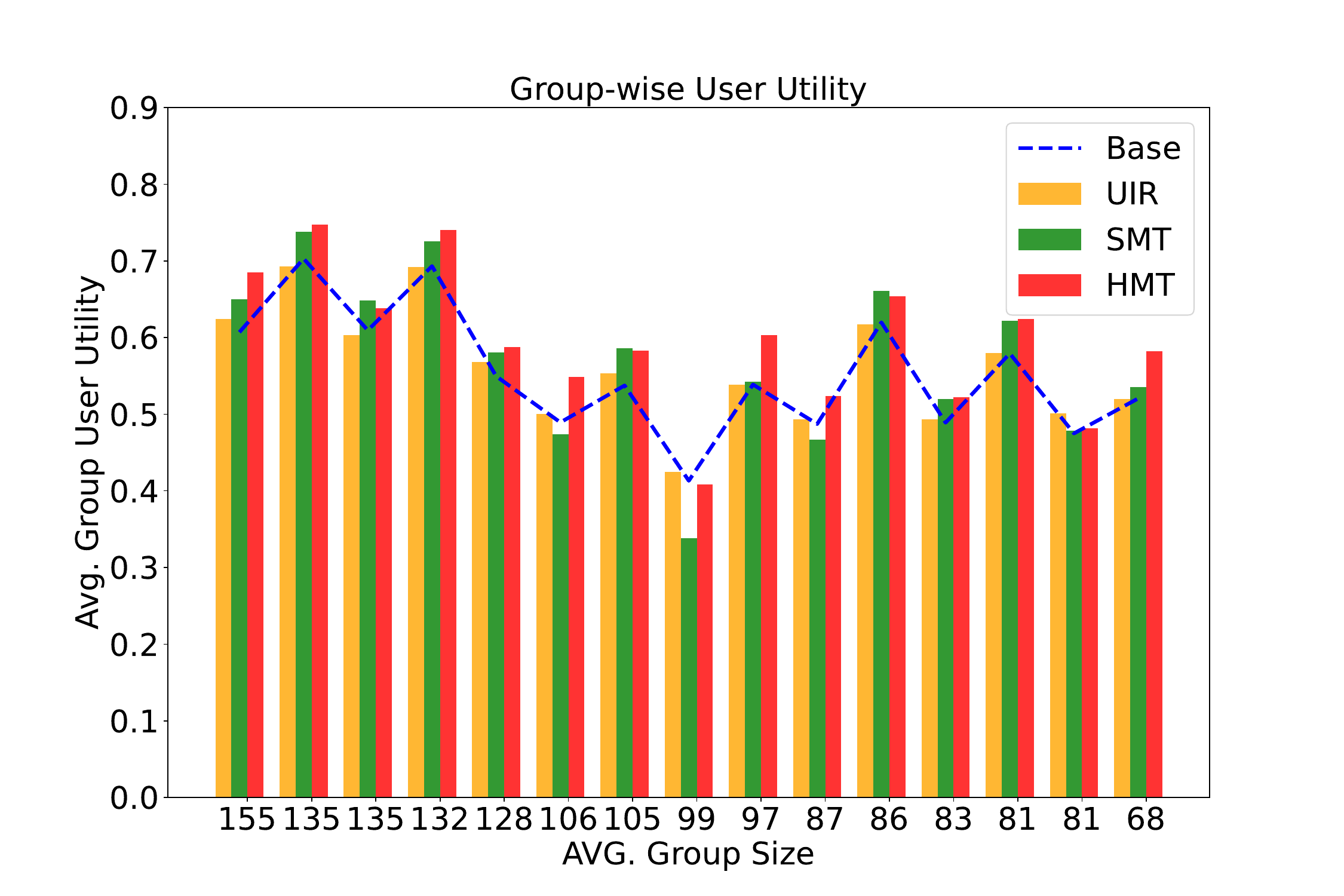}
\caption{Avg. group utilities.}
\label{fig:utility_group_2}
\end{subfigure}
\vspace{-3mm}
\caption{Performance of UIR, SMT and HMT on MovieLens-1m dataset against the no-intervention baseline. Results are averaged over 10 independent simulations including 0.2-sigma error bars. }
\label{fig:result_2}
\vspace{-3mm}
\end{figure*}

\subsubsection{Environment constructed from MovieLens-1m}

We use deep matrix factorization \citep{fan2018matrix} to train user and movie embeddings (with dimension set to $32$) by fitting the observed ratings in the range of 1 to 5. To ensure the quality of the trained embeddings, we performed a 5-fold cross-validation and obtained an averaged RMSE $=0.739$ on the test sets. Then with the same hyper-parameter settings, we train the user/item embeddings with the complete dataset. 

We select active users with more than 200 ratings, resulting in a population $\X$ comprising $1578$ users. We set the number of creators to $20$, with each creator's action set $\S_i$ consisting of $1000$ different movies. All $\{\S_i\}$ share a common part -- the most popular 700 movies based on the number of ratings they received, and each $\S_i$ also has a private part -- a randomly sampled 300 movies. Our choice of the user utility and matching score functions is $\pi(\s,\x)=\sigma(\s,\x)= \s^\top \x$, and then normalized to the region $[0,1].$ Additionally, we set $(\beta, K)=(0.1, 20)$ and initialize creators' strategies to the most preferred movie among all users (i.e., the movie that enjoys the highest average rating among $\X$). 

\subsubsection{Configurations of adaptive reweighting algorithm and intervention mechanisms}
For the adaptive reweighting algorithm, we set the epoch length $M=5$ and the simulation time horizon $T=3000$ for both environments. During each time step within an epoch, we simulate creators' responses by letting each of them update her strategy once using Algorithm \ref{alg:creator_dynamic} in a random order. Creators' learning rate is set to $\eta=0.2$. On the platform side, we use $K$-means clustering to determine user groups and set the number of clusters to $20$ for synthetic environment and $15$ for MovieLens environment, respectively. We should note as in practice, even the system does not have the exact knowledge about user distribution, we do not use the ground-truth clustering of users set in the simulation. In addition, we set the temperature parameter $\alpha=0.5$ for the first half of the time period and reduce it to $0.1$ for the remaining period. The clipping constants are set to $(w_{\min}, w_{\max})=(0.2, 5.0)$ and the mapping used in SMT and HMT are set to $\beta(\x)=\beta\cdot w(\x)$ and $K(\x)=\lceil K\cdot w(\x)\rceil$.


\subsubsection{Results}\label{sec:offline_result}

Figure \ref{fig:welfare_curve_1} illustrates the user welfare resulted from creators' evolving strategies under the three intervention mechanisms: UIR, SMT, and HMT, compared to the baseline (no platform intervention). Over time, all three mechanisms consistently outperform the baseline. In the baseline (shown in blue), the welfare plateaus quickly and remains stagnant. Conversely, the welfare curves under the other mechanisms exhibit ``double-ascent" patterns. Initially, they also plateau, but eventually, they begin to rise again and surpass the baseline. This is because, without platform's intervention, creators tend to remain in sub-optimal equilibria as illustrated in Section \ref{sec:failure}. However, our proposed mechanisms gradually accumulate user group weights, which, when significant enough, encourage creators to explore unattended user groups, leading to increased welfare. Among the three mechanisms, HMT demonstrates the most substantial gain with the least variance. UIR, while showing a lower marginal gain, maintains stability with minimal variance. SMT, which achieves a moderate gain, exhibits higher variance, suggesting that directly manipulating the matching temperature may be overly aggressive.

Figure \ref{fig:weight_group_1} shows the learned group weights at the last iteration of simulation. As it demonstrates, all three mechanisms emphasize on small groups over larger ones. This outcome aligns with our expectation: on one hand, larger user groups are more likely to ``trap'' unnecessarily many creators and thus should be deprioritized; on the other hand, increasing weights of niche user groups also improve their chances of being discovered by more creators. 

Figure \ref{fig:utility_group_1} breaks down the average utilities across user groups. The blue dashed line (i.e., the no-intervention baseline) exhibits a positive correlation between averaged group utility and group size, mirroring real-world observations. The orange bars show that UIR strikes a balance by improving the utility of niche groups while slightly trading off utility in larger groups. HMT achieves a remarkable Pareto improvement across all groups, as indicated by the red bars. However, SMT's gains come at the cost of even greater skewness in the average utility distribution across groups.

To summarize, all three mechanisms show promising improvements in overall user welfare, but their nature of gains differs, introducing considerations for the platform. When condition allows, HMT is the top choice due to its strong performance, stability, and fairness. For platforms that prioritize fairness and stability, UIR is also a viable option. However, SMT, despite improving overall welfare, may suffer from potential drawbacks such as instability and fairness issues. In-depth analysis of the merits and limitations of these mechanisms remains a topic for future research.

The results in the MovieLens environment align with the insights from the synthetic environment (refer to Figure \ref{fig:result_2}). However, it's worth noting that the trends in learned group weights and realized group utilities do not always align with group size, which is expected in real-world data where unattended user groups may not necessarily have small sizes. Nevertheless, our proposed mechanisms continue to improve overall welfare by identifying and prioritizing these groups.

\subsection{Online Experiments}

We conducted online evaluations on one of the world's leading short-video content creation and recommendation platforms (referred to as the ``platform'' hereafter due to the anonymity requirement), spanning over 3 weeks. We observe that the platform's intervention can indeed influence creator behavior because, on average, there is a positive correlation between the delivery volume and content creation volume for each topic. (The Pearson correlation is 0.2, and there are hundreds of topics in total). In this experiment, we employed the ``like-through-rate" (LTR) as the user utility function. LTR is calculated as the ratio of total likes to the number of impressions of a specific short video. We opted for LTR as the chosen metric because it not only serves as a reflection of user satisfaction but also offers a straightforward and easily interpretable signal for content creators to assess their content's perceived quality. The selection of the HMT mechanism for testing was deliberate, driven not only by its strong performance against the baseline and other mechanisms in our offline experiments, but also due to its ease of integration into production: HMT solely requires changing the number of candidate content retrieved for different users within the deployed relevance-based ranking model.

\subsubsection{Experiment Setups}  
We list the experiment setups below.

\noindent
{\bf User clustering:} We utilized explicit user characteristics such as demographics including country and gender and their level of activeness including video consumption volume and watch time. This approach led to the creation of over 10,000 user groups and we retained groups that had a sufficient number of users, resulting in hundreds of user groups. 

\noindent
{\bf Cluster weight update:} We implemented a daily weight updating cadence. Each day, we assessed the satisfaction of every user group by calculating the relative change of LTR over its average in the previous two days. Subsequently, we recalculated the user weights in accordance with the method outlined in Algorithm \ref{alg:ada_reweight}. 

\noindent
{\bf A/B test configurations:} To evaluate changes in both user and creator behavior, we employed a symmetric A/B test setup on the platform. This symmetric A/B test consisted of an experiment arm and a control arm to measure performance. At the beginning, we randomly pair 3\% creators with 3\% users from entire platform for each arm. Under this setup, users within each arm exclusively received content created by creators within the same arm, and content created by these creators was exclusively exposed to users within the same arm throughout the testing period. This stringent separation prevents any cross-group treatment leakage and maintains a closed feedback loop within each arm. In our online experiment, we ran these two arms for a duration of 3 weeks: a control arm adhering to the existing production setup and a test arm where we applied our proposed mechanism, HMT. 

\noindent
{\bf HMT specifics:} We implemented HMT during the cold start content retrieval phase, which pertains to content created within a few days and has not yet garnered a predefined number of impressions. Specifically, within the platform's production pipeline, we integrated an audience matching stage to retrieve cold start content. 
During this stage, content is exclusively delivered to the most suitable user candidates based on relevance scores generated by a pre-trained model. In the existing production setup, a fixed relevance score percentile of 99\% is uniformly applied to all users. This means that every user is only matched with the top 1\% of cold-start content in terms of relevance scores to ensure a high level of personalization. When tuning the percentile, we typically observe a trade-off between overall user satisfaction and the volume of cold-start content. In our experiment, we leveraged HMT to intelligently adjust this threshold for different user groups, anticipating improvements in both of thees metrics. Consequently, user groups with higher weights were granted a higher chance to be selected by content creators, while those with lower weights were deprioritized. The mapping from the group weight $w$ to the percentile of retrieved cold start content proportion was designed as a piece-wise constant function, with details specified in Table \ref{tab:w_mapping}.

\begin{table}
\centering
  \caption{Mapping $g$ in HMT}
  \label{tab:w_mapping}
  \begin{tabular}{cccccccc}
    \toprule
    Weight & $<1.0$ & $<1.19$ & $<1.79$ & $<2.13$ & $<2.36$ & $<2.68$ & $\geq2.68$\\
    \midrule
    Percentile & 0.99 & 0.95 & 0.90 & 0.85 & 0.75 & 0.7 & 0.1 \\
  \bottomrule
\end{tabular}
\end{table}

\subsubsection{Results}

Positive results were obtained in three key aspects. 

\noindent
{\bf User-side engagement:} The core utility metric LTR increased by $1.13\%$ and the total impression number of cold-start content increased by $0.76\%$, leading to a $3.7\%$ increase in impressions for fresh content created within 2 hours. These improvements are statistically significant and demonstrate increased user welfare while enhancing the freshness and diversity of content. The gains in both user satisfaction and the volume of cold-start content indicate that HMT influenced many creators to produce more targeted content that benefits niche user groups. Table \ref{tab:group_gain} provides a breakdown of performance improvement per user group. We indexed all groups in descending order by their sizes and divided them into four columns, with each column constituting approximately 25\% of the total traffic during the experiment period. As shown, smaller groups enjoyed a higher gain in terms of LTR, which echoes the observations in offline results. The gain in cold-start content impression volume shows an opposite trend. This is because the absolute number of cold-start impressions for larger user groups was smaller as the distribution of relevance scores in this group was more skewed, resulting in a larger relative gain in this metric. 

\noindent
{\bf Content diversity:} The average number of consumed topics per user during the experimental period increased by 0.71\%, and this increase is also statistically significant. 

\noindent
{\bf Creator-side engagement:} For popular creators (those with more than 1000 followers), the number of daily active users (Creator DAU) increased by an average of 0.17\%, while for the remaining creators, the gain is 0.06\%. Additionally, there is a promising increasing trend in Creator DAU for popular creators over the three weeks of the experiment: the increases over the first, second, and third weeks are -0.2\%, 0.24\%, and 0.48\%. This suggests that the three-week duration of the experiment may have been too short to influence the majority of creators to respond accordingly, and more time may be needed to fully observe the positive feedback from creators.

\begin{table}
\centering
  \caption{Gains per User Group}
  \label{tab:group_gain}
  \begin{tabular}{cccccccc}
    \toprule
    User Groups & 1-5 & 6-20 & 21-74 & 75+ & TOTAL \\
    \midrule
    LTR &        +0.43\% & +1.40\% & +0.75\% & +1.36\% & $\bm{+1.13\%}$  \\
    Impression & +2.64\% & +0.62\% & +1.42\% & +0.11\% & $\bm{+0.76\%}$ \\
  \bottomrule
\end{tabular}
\end{table}

\section{Conclusion}

In this study, we tackle the user welfare optimization challenge faced by online content recommendation platforms through the lens of mechanism design. We identified myopic strategy updates among creators caused by their limited information access as the culprit of sub-optimal welfare and introduced platform interventions to address this issue. Our three proposed mechanisms, based on adaptive user importance reweighting, enable platforms to convey global user preference information, reshape creators' perceived utilities, and influence their behaviors. Empirical experiments in both offline and online environments demonstrated the effectiveness of our approach, highlighting its potential for practical impact.

For future work, there remains an intriguing need for a comprehensive understanding of the merits and limitations of UIR, SMT, and HMT to aid practitioners in selecting the most suitable mechanism for real-world applications. It is also important to address practical constraints when applying the developed mechanisms. For instance, can we find ways to jointly optimize user welfare and platform costs? Can the mechanism explicitly ensure fairness on the user side and producer side? Deeper insights into these questions hold the potential to greatly impact the rapidly evolving online content landscape and industry practices.


\vskip 0.2in

\bibliography{main}

\appendix 
\newpage
\appendix

\section{Details of content creators' strategy update dynamics}\label{app:lbr}

The Local Better Response (LBR) procedure described in Algorithm \ref{alg:creator_dynamic} captures the evolution of creators' strategies in a snapshot, and characterizes two fundamental properties of content creation: 1. it relies solely on point estimations of the utility function (Line 3); and 2. it only incurs local changes at each update (Line 4). At each step, a creator who decides to update her strategy would first generate an exploration direction $\g_i$ (Line 2); then she would evaluate whether adjusting her strategy in this direction results in a higher utility. If so, she proceeds to update her strategy along $\g_i$ in a pace of $\eta$; otherwise, she maintains her current strategy. 

Algorithm \ref{alg:creator_dynamic} closely emulates real-world scenarios where creators strive to optimize their utilities while having merely black-box access to the utility functions. In practice, finding a clear direction that guarantees improved utility can be a challenging and, at times, unrealistic task. Consequently, we model their strategy evolution as an iterative process of trial and error. By definition, when LBR converges in  \cgame{}, it must converge to an LNE. Our primary interest lies in understanding how the platform can devise a dynamic rewarding or matching principle that maximizes cumulative user welfare within a given time period.
\vspace{-2mm}
\begin{algorithm}[h]
   \caption{({\bf LBR}) Local Better Response update at time step $t$}
   \label{alg:creator_dynamic}
\begin{algorithmic}[1]
   \STATE {\bfseries Input:} Learning rate $\eta$, an \cgame{} instance including utility functions and strategy sets $(u_i(\s), \S_i)$ of creator $i$, the joint strategy profile $\s^{(t)}=(\s_1^{(t)}, \cdots, \s_n^{(t)})$ at the current step $t$.
    \STATE Generate a random direction $\g_i \in \mathbb{S}^{d}$.
    \IF{ $u_i(\s_i^{(t)}+\eta \g_i,\s_{-i}^{(t)}) \geq u_i(\s^{(t)})$ }
        \STATE $\s_i^{(t+\frac{1}{2})}=\s_i^{(t)} + \eta \g_i$.
        \STATE Find $\s_i^{(t+1)}$ as the projection of $\s_i^{(t+\frac{1}{2})}$ in $\S_i$.
    \ELSE 
        \STATE $\s_i^{(t+1)}=\s_i^{(t)}$
    \ENDIF       
\end{algorithmic}
\end{algorithm}
\vspace{-3mm}

\section{Implementation Details of UIR, SMT and HMT Mechanisms}\label{app:implement_mechanisms}

The following sub-routine, denoted as Algorithm \ref{alg:sys_intervention_uir}, outlines how the platform deploys the weights obtained from Line 8, Algorithm \ref{alg:ada_reweight} as an intervention mechanism in Line 4. In Algorithm \ref{alg:sys_intervention_uir}, the weight vector $\w$ is directly employed to modify the reward or payment associated with each creator-user interaction. 

\begin{algorithm}[h]
   \caption{UIR Intervention}
   \label{alg:sys_intervention_uir}
\begin{algorithmic}
   \STATE {\bfseries Input:} Default recall capacity $K$, matching temperature $\beta$.
   \FOR{each user request $\x$ }
        \STATE Compute the relevance scores $\{\sigma(\s_i, \x)\}_{i=1}^n$.
        \STATE Retrieve the top-$K$ ranked content $\{s_{l(1)},\cdots,s_{l(K)}\}$ list based on relevance scores and randomly sample one element according to $\text{Softmax}(\{\beta^{-1}\sigma(\s_{l(i)}, \x)\}_{i=1}^K)$.
        \STATE For the user's choice $\s_i$, adjust creator-$i$'s default reward (payment) from $R(\s_i, \x)$ to $w(\x)R(\s_i, \x)$.
   \ENDFOR
\end{algorithmic}
\end{algorithm}

In the case of SMT or HMT intervention types, the platform requires a function to map $w(\x)$ to $\beta(\x)$ or $K(\x)$. This mapping can be implemented as a piecewise constant function and determined empirically. The specifics of this process are elucidated in Algorithm \ref{alg:sys_intervention_smt} and \ref{alg:sys_intervention_hmt}.

\begin{algorithm}[h]
   \caption{SMT Intervention}
   \label{alg:sys_intervention_smt}
\begin{algorithmic}
   \STATE {\bfseries Input:} Default recall capacity $K$, matching temperature $\beta$, $f:\RR_{+}\rightarrow \RR_{+}$.
   \FOR{each user request $\x$ }
        \STATE Compute the relevance scores $\{\sigma(\s_i, \x)\}_{i=1}^n$.
        \STATE Retrieve the top-$K$ ranked content $\{s_{l(1)},\cdots,s_{l(K)}\}$ list based on relevance scores and randomly sample one element according to $\text{Softmax}(\{\beta(\x)^{-1}\sigma(\s_{l(i)}, \x)\}_{i=1}^K)$, where $\beta(\x)=f(w(\x))$.
   \ENDFOR
\end{algorithmic}
\end{algorithm}

\begin{algorithm}[h]
   \caption{HMT Intervention}
   \label{alg:sys_intervention_hmt}
\begin{algorithmic}
   \STATE {\bfseries Input:} Default recall capacity $K$, matching temperature $\beta$, $g:\RR_{+}\rightarrow \NN_{+}$.
   \FOR{each user request $\x$ }
        \STATE Compute the relevance scores $\{\sigma(\s_i, \x)\}_{i=1}^n$.
        \STATE Retrieve the top-$K(\x)$ ranked content $\{s_{l(1)},\cdots,s_{l(K(\x))}\}$ list based on relevance scores and randomly sample one element according to $\text{Softmax}(\{\beta^{-1}\sigma(\s_{l(i)}, \x)\}_{i=1}^K(\x))$, where $K(\x)=g(w(\x))$.
   \ENDFOR
\end{algorithmic}
\end{algorithm}

\section{Proof of Theorem \ref{thm:PNE}}\label{app:proof_pne}

We restate Theorem \ref{thm:PNE} as the following with more rigorous characterizations, and then provide its detailed proof. 
\begin{theorem}
    Any \cgame{} game with $K=n$ has a unique pure Nash equilibrium (PNE) if each creator's srtategy set $\S_i$ is convex and $\sigma(\cdot,\x)$ is twice-differentiable and satisfies 
   \begin{equation}\label{eq:suff_monotone}
       \EE_{\x\sim\X}\left[\frac{\partial^2 \sigma}{\partial \s_i^2}+\Big(\frac{\partial \sigma}{\partial \s_i}\Big)\Big(\frac{\partial \sigma}{\partial \s_i}\Big)^{\top}\right] \preceq 0, \forall i\in[n].
   \end{equation}
\end{theorem}

\begin{proof}

We prove that under the proposed conditions, the \cgame{} is a strictly monotone game \cite{rosen1965existence} and thus possesses a unique PNE.
According to Appendix A in \cite{bravo2018bandit}, a sufficient condition that establishes strictly monotonicity for any $n$-person game $\G$ is convex action sets and a negative definite Hessian $[H^{\G}_{ij}]$ of $\G$, which is defined as 
    $$H_{ij}(\s)=\frac{1}{2}\nabla_j\nabla_i u_i(\s)+\frac{1}{2}\nabla_i\nabla_j u_j(\s)^{\top}.$$
    
For \cgame{} game, the convexity of strategy sets are satisfied. Next we prove the property of the game's Hessian matrix with associated utility function 
\begin{equation}
    u_i(\s)=\EE_{\x\sim\X} \left[\frac{\exp(\sigma(\s_i, \x))}{\sum_{l=1}^n \exp(\sigma(\s_l, \x))}\right].
\end{equation}
    
Without loss of generality, let $\beta=1$. Denote $A_i=\exp(\sigma(\s_i, \x)), M=A_1+\cdots+A_n$, we have

\begin{align*}
    H_{ii} &= -\EE_{\x\sim\X}\Big\{\Big[-\frac{\partial^2 \sigma}{\partial \s_i^2}-\Big(\frac{\partial \sigma}{\partial \s_i}\Big)\Big(\frac{\partial \sigma}{\partial \s_i}\Big)^{\top}\Big]A_i(M-A_i)\cdot\frac{1}{M^2} \\ &\quad\quad\quad\quad\quad+\frac{2}{M}\Big(\frac{\partial \sigma}{\partial \s_i}\Big)\Big(\frac{\partial \sigma}{\partial \s_i}\Big)^{\top}A_i^2(M-A_i)\cdot \frac{1}{M^2}\Big\} \\ 
    & =-\EE_{\x\sim\X}\Big\{\Big[-\frac{\partial^2 \sigma}{\partial \s_i^2}-\Big(\frac{\partial \sigma}{\partial \s_i}\Big)\Big(\frac{\partial \sigma}{\partial \s_i}\Big)^{\top}(1-\frac{A_i}{M})\Big]\cdot A_i(M-A_i)\frac{1}{M^2}\Big\} \\ &\quad- \EE_{\x\sim\X}\Big\{\Big(\frac{\partial \sigma}{\partial \s_i}\Big)\Big(\frac{\partial \sigma}{\partial \s_i}\Big)^{\top}A_i^2(M-A_i)\cdot \frac{1}{M^3}\Big\} \\
    & \triangleq -\EE_{\x\sim\X}\left[H_{ii}^{(0)}(\s,\x)\right]-\EE_{\x\sim\X}\left[H_{ii}^{(1)}(\s,\x)\frac{1}{M^3}\right].
    \end{align*}

    \begin{align*}
        H_{ij} &= -\EE_{\x\sim\X}\Big\{\Big(\frac{\partial \sigma}{\partial \s_i}\Big)\Big(\frac{\partial \sigma}{\partial \s_j}\Big)^{\top}A_iA_j(M-A_i-A_j)\cdot \frac{1}{M^3} \Big\}\\ 
        & \triangleq - \EE_{\x\sim\X}\left[H_{ij}^{(1)}(\s,\x)\frac{1}{M^3}\right].
    \end{align*}
    
    Next we show that for any $\x$ and $\s$, the block matrix $[H^{(1)}_{ij}]$ is always positive semi-definite (PSD). For simplicity, let
    \begin{align*}
        & \y_i=A_i\frac{\partial \sigma}{\partial \s_i}\in \RR^{d\times 1}, \y=[\y_1;\dots;\y_n]\in\RR^{dn\times 1}, \\ 
        &\z=[A_1\y_1;\dots;A_n\y_n]\in\RR^{dn\times 1},
    \end{align*}
we obtain
\begin{small}
    \begin{align*}
        &[H^{(1)}_{ij}]= \\ &\left [
    \begin{array}{cccc}
        \y_1\y_1^{\top}(M-A_1) & \y_1\y_2^{\top}(M-A_1-A_2) & \cdots & \y_1\y_n^{\top}(M-A_1-A_n) \\
        \y_2\y_1^{\top}(M-A_2-A_1) & \y_2\y_2^{\top}(M-A_2) & \cdots & \y_2\y_n^{\top}(M-A_2-A_n) \\
        \vdots & \vdots & \ddots & \vdots \\
        \y_n\y_1^{\top}(M-A_n-A_1) & \y_n\y_2^{\top}(M-A_n-A_2) & \cdots & \y_n\y_n^{\top}(M-A_n) \\
    \end{array}\right ] \\ 
     = &M\y\y^{\top}-\y\z^{\top}-\z\y^{\top} + \text{diag}(A_1\y_1\y_1^{\top},\dots,A_n\y_n\y_n^{\top}) \\ 
     = &\frac{1}{M}\cdot (M\y-\z)(M\y-\z)^{\top} + \text{diag}(A_1\y_1\y_1^{\top},\dots,A_n\y_n\y_n^{\top})-\frac{1}{M}\z\z^{\top} \\
     \succ &\text{diag}(A_1\y_1\y_1^{\top},\dots,A_n\y_n\y_n^{\top})-\frac{1}{M}\z\z^{\top}.
    \end{align*}
 \end{small}    
    Therefore, it suffices to prove that the matrix $$\tilde{H}=M\text{diag}(A_1\y_1\y_1^{\top},\dots,A_n\y_n\y_n^{\top})-\z\z^{\top}$$ is PSD. For any $\t=[\t_1;\cdots; \t_n]\in\RR^{dn\times 1}$ where $\t_i\in\RR^d$, we can verify that
    \begin{align*}
        \t^{\top}\tilde{H}\t &= M\sum_{i=1}^n A_i\t_i^{\top}\y_i\y_i^{\top} \t_i - \t^{\top}\z\z^{\top} \t\\ 
        & =\sum_{i=1}^n A_i\sum_{i=1}^n A_i\t_i^{\top}\y_i\y_i^{\top} \t_i - \t^{\top}\z\z^{\top} \t \\
        &= \sum_{1\leq i< j\leq n}A_iA_j(\y_i^{\top}\t_i-\y_j^{\top}\t_j)^2\geq 0.
    \end{align*}
    
    Therefore, the block matrix $[H^{(1)}_{ij}]$ is always PSD for any $\x$ and $\s$. A sufficient condition for $[H_{ij}^{\cG}]$ to be negative definite is thus $H_{ii}^{(0)}$ being positive definite (PD), i.e., $H_{ii}^{(0)}(\s,\x) \succ 0, \forall \s,\x$. It remains to show that
    \begin{equation}\label{eq:57}
        \EE_{\x\sim\X}\left[\Big[-\frac{\partial^2 \sigma}{\partial \s_i^2}-\Big(\frac{\partial \sigma}{\partial \s_i}\Big)\Big(\frac{\partial \sigma}{\partial \s_i}\Big)^{\top}(1-\frac{A_i}{M})\Big]\cdot A_i(M-A_i)\frac{1}{M^2}\right] \succ 0.
    \end{equation}
    And a sufficient condition for Eq \eqref{eq:57} to hold is
    \begin{equation}
      \EE_{\x\sim\X}\left[-\frac{\partial^2 \sigma}{\partial \s_i^2}-\Big(\frac{\partial \sigma}{\partial \s_i}\Big)\Big(\frac{\partial \sigma}{\partial \s_i}\Big)^{\top}\right] \succeq 0,
    \end{equation}
    which completes the proof.

\end{proof}

\section{Proof of Theorem \ref{thm:opt}}
\begin{proof}
    Since the utility functions of \cgame{} are twice differentiable, any LNE $\s$ of \cgame{} satisfies the following definition
\begin{equation}
    \s_{i} = \arg\max_{\z_i\in B(\s_i, \delta)} u_i(\z_i, \s_{-i};\w)
\end{equation}
must also satisfy the first-order condition
$\frac{\partial u_i}{\partial \s_i}\Big|_{\s=(\s_i,\s_{-i)}} = 0$. If we let 
\begin{equation}
    F(\s, \w)=\left(\frac{\partial u_1(\s;\w)}{\partial \s_1}, \cdots, \frac{\partial u_n(\s;\w)}{\partial \s_n}\right): \RR^{dn}\times \RR^{m}_{\geq 0}\rightarrow \RR^{dn}
\end{equation}
be a vector-valued function, the constraint \eqref{eq:bi-level-opt-inner} can be rewritten into 
\begin{equation}
    F(\s^*(\w), \w)=0.
\end{equation}

From the implicit function theorem \citep{krantz2002implicit}, the derivative of $\s^*$ w.r.t. $\w$ can be written as

$$\frac{d \s}{d\w} = -\left(\frac{\partial F}{\partial \s}\right)^{-1} \cdot \frac{\partial F}{\partial \w}, $$
where $\left[\frac{\partial F}{\partial \s}\right]_{nd\times nd}, \left[\frac{\partial F}{\partial \w}\right]_{nd\times m}$ are the Jacobian matrices, and 

\begin{equation}\label{eq:UIR_gradient}
    \frac{d W}{d \w}=\frac{d W}{d \s}\cdot \frac{d\s}{d \w}=-\frac{d W}{d \s}\cdot \left(\frac{\partial F}{\partial \s}\right)^{-1} \cdot \frac{\partial F}{\partial \w},
\end{equation}

where $\left(\frac{d W}{d \s}\right)_{1\times nd}$ is the partial derivative of $W$ w.r.t. $\s$.

Since $w_j\geq 0$, we apply a change of variable and denote each $w_j$ as $e^{w_j}$ instead. Next we calculate each term of the RHS of \eqref{eq:UIR_gradient} to obtain an estimation of the gradient of our objective welfare function $W$ to the user weight vector $\w$. Without loss of generality we let the user distribution $\X$ be a uniform distribution on unit basis $\{\e_1,\cdots,\e_d\}$ and $m=d$. The utility functions given in Eq \eqref{eq:creator_expected_utility_func2} and the user welfare function read

    \begin{equation}
        W(\s) =  \frac{1}{m}\sum_{j=1}^d \sum_{i=1}^n \s_i^{\top}\x_j\cdot \frac{\exp[\beta^{-1} \s_i^{\top}\x_j]}{\sum_{k=1}^n \exp[\beta^{-1} \s_k^{\top}\x_j]} .
    \end{equation}
    \begin{equation}
       u_i(\s_i,\s_{-i})=\frac{1}{m}\sum_{j=1}^d e^{w_j}\cdot\frac{\exp[\beta^{-1} \s_i^{\top}\x_j]}{\sum_{k=1}^n \exp[\beta^{-1} \s_k^{\top}\x_j]} , i\in[K]. 
    \end{equation}

If we denote $A_{ij}=\exp[\beta^{-1} \s_i^{\top}\x_j], M_j=\sum_{k=1}^n \exp[\beta^{-1} \s_k^{\top}\x_j]$, then $\frac{A_{ij}}{M_j}=P_i(\s,\x_j)$ is exactly the probability of matching content $\s_i$ to $\x_j$. Given the assumption that $n$ is sufficiently large, we have $\frac{A_{ij}}{M_j}=o(1)$ is sufficiently small for any $i$ and therefore we ignore the high-order infiintesimal terms such as $\frac{A_{ij}^2}{M_j^2}, \frac{A_{kj}A_{ij}}{M_j^2}$ in the following derivation.

\begin{align}\notag
    \frac{d W}{d \s_i}=&\frac{1}{m}\sum_{j=1}^d \x_j \left[ \frac{A_{ij}}{M_j} + \beta^{-1}\s_i^{\top}\x_j\left(\frac{A_{ij}}{M_j}-\frac{A_{ij}^2}{M_j^2}\right)\right] \\ \notag &-\frac{1}{m}\sum_{j=1}^d \x_j \left[\sum_{k\neq i} \beta^{-1}\s_k^{\top}\x_j \frac{A_{kj}A_{ij}}{M_j^2}\right] \\ 
    \approx & \frac{1}{m}\sum_{j=1}^d \x_j\frac{A_{ij}}{M_j} \left( 1 + \beta^{-1}\s_i^{\top}\x_j\right), i \in [n],
\end{align}
where $\bar{\pi}(\x_j)\triangleq \sum_{k=1}^n \s_k^{\top}\x_j\frac{A_{kj}}{M_j}$.

Next we calculate each term in the RHS of Eq \eqref{eq:UIR_gradient}. The $i$-th block of $F(\s,\w)$ is a $d$-dimensional vector given by

\begin{align}\notag
    F(\s,\w)_i=&\frac{1}{m}\sum_{j=1}^d \x_j \left[  \beta^{-1}e^{w_j}\left(\frac{A_{ij}}{M_j}-\frac{A_{ij}^2}{M_j^2}\right)\right] \\ 
    \approx & \frac{1}{m}\sum_{j=1}^d \x_j\frac{A_{ij}}{M_j} \beta^{-1}e^{w_j}, i \in [n],
\end{align}
the $(i,j)$-th block of $\frac{\partial F}{\partial \w}$ is a $d$-dimensional vector given by
\begin{align}\notag
    \left[\frac{\partial F}{\partial \w}\right]_{ij}=& \frac{1}{m}\x_j \beta^{-1}e^{w_j}\left(\frac{A_{ij}}{M_j}-\frac{A_{ij}^2}{M_j^2}\right) \\ 
    \approx & \frac{1}{m}\x_j \beta^{-1}\frac{A_{ij}}{M_j}e^{w_j}, i \in [n], j\in[d].
\end{align}
Since $\{\x_i\}_{i=1}^n$ are orthogonal basis, the non-diagonal blocks of matrix $\frac{\partial F}{\partial \s}$ are all zero matrices and the $i$-th diagonal block of matrix $\frac{\partial F}{\partial \s}$ is given by 
\begin{align}
    \left[\frac{\partial F}{\partial \s}\right]_{ii}&= \frac{1}{m\beta^2}\sum_{j=1}^de^{w_j}\left[\x_j\x_j^{\top}\left(\frac{A_{ij}}{M_j}-\frac{3A_{ij}^2}{M_j^2}+\frac{2A_{ij}^3}{M_j^3}\right) \right] \\ \notag
    & \approx \frac{1}{m\beta^2}\sum_{j=1}^de^{w_j}\left[\x_j\x_j^{\top}\frac{A_{ij}}{M_j} \right], i\in[n].
\end{align}

Therefore, we can derive a approximation of $\frac{dW}{d\w}$ as below:

\begin{align}\notag
    \frac{d W}{d w_j}&=-\frac{d W}{d \s}\cdot \left(\frac{\partial F}{\partial \s}\right)^{-1} \cdot \left(\frac{\partial F}{\partial \w}\right)_j \\\notag
    & \approx -\sum_{i=1}^n \Bigg\{ \frac{1}{m}\sum_{k=1}^d \x^{\top}_k\frac{A_{ik}}{M_k} \left( 1 + \beta^{-1}\s_i^{\top}\x_k\right) \\ \notag
    &\cdot m\beta^2\text{diag}^{-1}(e^{w_1}A_{i1}/M_1, \cdots, e^{w_L}A_{id}/M_d)  \cdot\frac{1}{m}\x_j \beta^{-1}\frac{A_{ij}}{M_j}e^{w_j} \Bigg\} \\ \notag
    & = -\frac{\beta^2}{m}\sum_{i=1}^n e^{-w_j}\left(1+\beta^{-1}\s_i^{\top}\x_j\right)\beta^{-1}\frac{A_{ij}}{M_j}e^{w_j} \\ \label{eq:72}
    & \approx -\frac{1}{m }\sum_{i=1}^n \s_i^{\top}\x_j \frac{A_{ij}}{M_j} \\ \notag
    & =-\frac{1}{m} \sum_{i=1}^n \pi(\s_i, \x_j) P_i(\s, \x_j)  \\ \label{eq:75}
    & = -\frac{1}{m} \EE_{\s}[\pi(\x_j)],
\end{align}
where \eqref{eq:72} holds because $\beta^{-1}>>1$.

Therefore, Eq \eqref{eq:75} suggests that the following update rule
\begin{equation}
    e^{w'_j} = e^{w_j}\cdot e^{-\eta \bar{\pi}(\x_j)}
\end{equation}
aligns with the gradient direction of $W(\w)$,  
which yields Eq \eqref{eq:84}.

\end{proof}

\end{document}